\documentclass[a4paper,11pt]{article}

\usepackage{tikz}
\usetikzlibrary{arrows.meta,chains,decorations.pathreplacing}

\usepackage{hyperref}
\hypersetup{colorlinks=true,linkcolor=blue,citecolor=blue,urlcolor=blue}

\usepackage{fullpage} 

\usepackage{color-edits}
\usepackage{algorithm}
\usepackage{bbm}
\usepackage{thm-restate}

\addauthor{mb}{red}
\addauthor{uf}{blue}
\addauthor{te}{brown}

\newcommand{\mbc}[1]{{\mbcomment{#1}}}

\usepackage{amsmath}
\usepackage{amssymb}
\usepackage{xcolor}
\usepackage{booktabs}

\usepackage[numbers]{natbib} 

\newtheorem{theorem}{Theorem}
\newtheorem{corollary}[theorem]{Corollary}

\newtheorem{lemma}[theorem]{Lemma}
\newtheorem{claim}[theorem]{Claim}

\newtheorem{proposition}[theorem]{Proposition}
\newtheorem{observation}[theorem]{Observation}
\newtheorem{remark}[theorem]{Remark}
\newtheorem{example}[theorem]{Example}

\newenvironment{proof}{\noindent\bf{Proof.}\rm}{\hfill$\blacksquare$\bigskip} 

\newcommand{\cout}[1]{}
\newcommand{\mech}[0]{\mathcal{M}}
\newcommand{\mechx}{\mech^X}
\newcommand{\alx}{A^X}
\newcommand{\alxa}{\tilde{A}^X}
\newcommand{\epsap}{\frac{1}{1+\epsilon}}

\newcommand{\mechb}{\mech^B}
\newcommand{\meche}{\mech^L}
\newcommand{\alb}{A^B}
\newcommand{\alba}{\tilde{A}^B}

\newcommand{\groundset}{{\mathcal{U}}}
\newcommand{\indsets}{{\mathcal{I}}}

\newcommand{\meps}{m\epsilon}

\title{Fair and Truthful  Mechanisms for Dichotomous Valuations}

\author{
	Moshe Babaioff\thanks{
	Microsoft Research. \texttt{moshe@microsoft.com}.} \and 
	Tomer Ezra\thanks{Computer Science, Tel-Aviv University. \texttt{tomer.ezra@gmail.com}.}
	\and
	Uriel Feige\thanks{Weizmann Institute, Rehovot, Israel, \texttt{Uriel.Feige@weizmann.ac.il}. Supported in part by the Israel Science Foundation (grant No. 1388/16). Part of this work was done at Microsoft Research, Herzeliya.}
}

\begin{document}	



\maketitle
\begin{abstract}
We consider the problem of allocating a set on indivisible items to agents with private preferences in an efficient and fair way.
We focus on valuations that have \emph{dichotomous marginals}, in which the added value of any item to a set is either 0 or 1, and aim to design  truthful allocation mechanisms (without money) that maximize welfare and are fair.
For the case that agents have submodular valuations with dichotomous marginals, we design such a deterministic truthful allocation mechanism. The allocation output by our mechanism is Lorenz dominating, and consequently satisfies many desired fairness properties, such as being envy-free up to any item (EFX), and maximizing the Nash Social Welfare (NSW). 
We then show that our mechanism with random priorities is envy-free ex-ante, while having all the above properties ex-post. Furthermore, we present several impossibility results precluding similar results for the larger class of XOS valuations.  


To gauge the robustness of our positive results, we also study $\epsilon$-dichotomous valuations, in which the added value of any item to a set is either non-positive, or in the range $[1, 1 + \epsilon]$. We show several impossibility results in this setting, and also a positive result: for agents that have additive $\epsilon$-dichotomous valuations with sufficiently small $\epsilon$, we design a randomized truthful mechanism with strong ex-post guarantees. For $\rho = \frac{1}{1 + \epsilon}$, the allocations that it produces generate at least a $\rho$-fraction of the maximum welfare, and enjoy $\rho$-approximations for various fairness properties, such as being envy-free up to one item (EF1), and giving each agent at least her maximin share.
\end{abstract}





\section{Introduction}

A central problem in Algorithmic Game Theory is the problem of allocating indivisible goods among agents with private preferences. This problem is particularly challenging in settings in which utilities cannot be transferred between agents (no money). One consideration in allocating the items is the economic efficiency of the allocation, as we want the best for society as a whole. Another consideration is fairness of the allocation, because in the absence of money, there is no other way for the agents to evenly share the welfare generated by the efficient allocation.

Thus, in this work we design allocation mechanisms that enjoy desirable properties, related to their economic efficiency, to fairness of the allocation, and to incentive compatibility (truthfulness). Importantly, we consider only settings without money, so a mechanism defines an \emph{allocation rule}, but does not involve a payment rule, as there are no payments. 
With general valuations, even without any fairness properties, the VCG mechanism is the unique truthful welfare-maximizing mechanism, and it requires payments. Consequently, the focus of our work is on instances in which the valuation functions of the agents are restricted, and specifically, have the {\em dichotomous marginals} property.
We say that a valuation function $f$ has {\em dichotomous marginals} (or for brevity, we simply say that $f$ is {\em dichotomous}) if for every set $S$ of items and every additional item $a$, the marginal value of $a$ relative to $S$ is either~0 or~1. Namely, $f(S \cup a) - f(S) \in \{0,1\}$.

{The study of fairness with dichotomous preferences was initiated by Bogomolnaia and Moulin~\cite{BM2004}, with additional extensive research of such preferences in various settings (see e.g. ~\cite{BMS2005, RuthSU2005, Freitas2010, BouveretL08, KPS2018, Ortega2020}).}
{The above references provide multiple examples of situations that can be modeled using dichotomous preferences. Next we provide another example that involves constraints not captured by prior work. Consider a setting where the agents are arts students} seeking work as museum guides. The items are different shifts in which the students can work as guides in the local arts museum. Suppose that among the shifts (or combinations of shifts) that are feasible for a given student in a given month (for example, one student cannot work on weekends, another student can work at most two shifts a week, etc.), the student may wish to work for as many shifts as possible during the month, but other than that is indifferent to the exact choice of shifts (as long as the combination of shifts is feasible for the student).
{A model that first-order approximates this setting is one in which the valuation function of a student is modeled as being dichotomous.\footnote{{We later discuss relaxing the assumption that all desired shifts are equivalent for the student, allowing strict preference, while still assuming \emph{approximately} the same marginals.}}} 
The allocation problem is to assign students to shifts. Economic efficiency may correspond to filling as many shifts as possible. 
Fairness may correspond to trying to equalize the number of shifts that each student receives (subject to the feasibility constraints). Incentive compatibility means that it is a dominant strategy for a student to report her true valuation function 
to the museum  -- providing an incorrect report cannot lead to a situation in which she receives a bundle of shifts of higher value to her.

\subsection{Our Contribution and Techniques}
We now provide an overview of our main results. Some definitions and technicalities are omitted from this overview, but can be found in Section \ref{sec:model}. 

We consider settings with a finite set $M$ of $m$ indivisible and non-identical items.
There is a set of $n \ge 2$ agents (a.k.a. {\em agents}), denoted by  $V=[n]$, with each agent $v\in V$ having a \emph{valuation function} $f_v$ over sets of items.
The \emph{value} (or utility) of agent $v$ for a set $S\subseteq M$ is denoted by $f_v(S)$.
We always assume that any valuation $f$ is normalized ($f(\emptyset) = 0$) and non-decreasing ($f(S)\leq f(T)$ for $S\subseteq T \subseteq M$).
Given an allocation $A$, we use $A_v$ to denote the set of items allocated to agent $v$.

One question that we ask in this work is what is the largest class of dichotomous valuation functions for which one has a truthful deterministic allocation mechanism that enjoys good economic efficiency and fairness properties. Before presenting our results, let us briefly discuss its various ingredients.

{\bf Classes of valuation functions.} The dichotomous versions of some simple classes of valuation functions were considered in previous work (e.g., {\em unit demand} {(matching)~\cite{BM2004}, additive~\cite{Ortega2020} {and $0/1$ valued sets~\cite{BouveretL08}})}. 
We consider here the hierarchy of complement-free valuation functions introduced in~\cite{LLN}, whose four highest classes (in order of containment) are {\em gross substitutes} (GS), {\em submodular}, {\em XOS}, and {\em subadditive} (recall that both {\em unit demand} and {\em additive} are gross substitutes).
For valuations with dichotomous marginals, it can be shown that every submodular function is in fact a Matroid Rank Function (MRF), and hence also gross substitutes. 
We note that valuation functions may be used to express not only the preferences of the agents, but also constraints imposed by the allocator. In the museum example above, the museum may impose a restriction that no student can work in two shifts in the same day, and another restriction that no student can work in five shifts in the same week. If a student has an additive valuation function, then incorporating these constraints into her valuation function makes it submodular.

{\bf Economic efficiency.} We wish our allocations to maximize welfare, where the welfare of allocation $A$ is defined as $\sum_{v \in V} f_v(A_v)$. Restricting attention to non-redundant allocations (no item can be removed from {a set allocated to an agent without decreasing its value}), in the setting of dichotomous valuations, maximizing welfare is equivalent to allocating the maximum possible number of items.
Hence maximizing welfare can serve as a measure of economic efficiency not only from the point of view of the agents, but also from the point of view of the items (as in the museum guides example, where it is in the interest of the museum to fill as many shifts as possible). 

{\bf Fairness.} For allocation mechanisms without money it is customary to impose some fairness requirements. They come in many flavors. {\em Safety} guarantees (such as {\em proportionality}, {\em maximin share}\footnote{The maximin share of agent $p$ is the maximum value that could be given to the least happy agent if all agents had the same valuation function as that of $p$.}) promise the agent a certain minimum value, based only on the valuation function of the given agent and no matter what the valuation functions of other agents are. {\em Envy-freeness} guarantees ({\em envy free up to one good} (EF1), {\em envy free up to any good} (EFX)), ensure that every agent $v$ is at least as happy with her own bundle of goods as she would be with the bundle received by any other agent ({perhaps up to one good (EF1), or up to any good (EFX)}). {\em Egalitarian} guarantees ({\em lexicographically maximal} allocations, {\em Lorenz-dominating} allocations, maximizing Nash social welfare (NSW)) attempt to equalize the utilities of all agents (to the extent possible, given their valuation functions). Not all fairness notions are attainable in all settings, and in addition, there are settings in which two fairness notions that are attainable are not attainable simultaneously. For this reason, in our work we do not fix one particular fairness notion, but rather attempt to achieve a good mix of fairness properties.

{\bf Truthfulness.} We wish our mechanisms to have the property that reporting her true valuation function is a (weakly) dominant strategy for every agent. That is, for every agent $v$, whatever the reports of other agents are, if agent $v$ reports a valuation function different than $f_v$, the allocation she gets cannot have higher value to her, compared to the allocation when she reports $f_v$.

We now return to our question concerning the largest class of dichotomous valuation functions for which one has a truthful deterministic allocation mechanism that enjoys good economic efficiency and fairness properties.
We address this question in the framework of the hierarchy of complement-free valuation functions
defined in~\cite{LLN}. Our {first} main result shows that if the dichotomous valuation functions are submodular, then a deterministic mechanism that we refer to as {\em prioritized egalitarian} (PE) indeed satisfies the above requirements.

\begin{theorem}
	\label{thm:submodular-main-intro}
	The {\em prioritized egalitarian} (PE) mechanism has the following properties when agents have submodular dichotomous valuations:
	
	\begin{enumerate}
		\item \label{item:submodular-truthful} Being truthful is a dominant strategy.
		
		\item \label{item:submodular-welfare} When agents are truthful the allocation is welfare maximizing.
		
		\item \label{item:submodular-fair}  When agents are truthful, the allocation of the mechanism is a Lorenz dominating allocation, and consequently it enjoys additional fairness properties, including maximizing Nash social welfare, and being envy-free up to any item (EFX). If furthermore, the valuations are additive dichotomous, the allocation gives every agent at least her maximin {share.} 
		
		\item \label{item:submodular-poly}  If the valuations of agents have succinct representations\footnote{As a convention, throughout this paper we assume that a succinct representation allows the following: 1) computation of function values in polynomial time, and 2) verifying in polynomial time that the succinct representation indeed corresponds to a submodular dichotomous valuation.  }, then the mechanism can be implemented in polynomial time.\footnote{In the value queries model, our polynomial-time implementation is ex-post incentive compatible.}
		
		
	\end{enumerate}
	
\end{theorem}

In contrast, {we show that} if the valuation functions belong to the class XOS (one level higher than submodular in the hierarchy of~\cite{LLN}), then there is no truthful allocation mechanism (neither deterministic nor randomized) that maximizes welfare, {even if one disregards all fairness considerations}. For this and some other impossibility results see  Appendix~\ref{sec:XOS}.

The PE mechanism is based on first proving that in the setting of submodular dichotomous valuation functions there always is a {\em Lorenz dominating} allocation 
(exact definitions will follow, but at this point the reader may think a Lorenz dominating allocation as one that both maximizes welfare and equalizes as much as possible the number of items received by each agent). The PE mechanism imposes a priority order $\sigma$ among agents, and chooses a non-redundant\footnote{
In Appendix~\ref{app:all} we discuss the issue of non-redundancy, showing that the result of Theorem  \ref{thm:submodular-main-intro} is impossible to obtain when one insists on allocating all items (even undesired ones).
	Specifically, we show that there is no truthful deterministic allocation mechanism that always allocates all items, maximizes welfare and is EFX. This holds even for additive dichotomous valuations, and even for only two agents.}
 Lorenz dominating allocation (namely, it does not allocate items that give~0 marginal value to the agent receiving them), breaking ties among Lorenz dominating allocations in favor of higher priority agents. Proving economic efficiency and fairness properties for this mechanism is straightforward, given the fact that the output allocation is Lorenz dominating. The main technical content in the proof of Theorem~\ref{thm:submodular-main-intro} (beyond the proof that a Lorenz dominating allocation exists) is to show that the PE mechanism is truthful (for agents with submodular dichotomous valuations).
 
 In Section~\ref{sec:rpe} we consider a randomized variation of our PE allocation mechanism. This randomized mechanism first assigns the agents priorities uniformly at random, and then runs the PE allocation mechanism with the drawn priorities. We show that this mechanism achieves Envy-Freeness in expectation (ex-ante), is universally truthful and it obtains all the other good properties of the PE mechanism ex-post (a best-of-both-worlds result).


Armed with the above results for dichotomous valuations, we study whether our positive results are robust in face of slight violations of the dichotomous assumption. For simplicity of the presentation, consider the special case of additive dichotomous valuations. The dichotomous assumption models situations in which items of value~1 are ``desirable" whereas items of value~0 are not desirable, and an agent is indifferent among items that she finds desirable (and likewise, indifferent among items that she finds undesirable). A natural relaxation for the undesirable items is to allow them to have arbitrary non-positive value. It turns out that the PE mechanism (and other natural mechanisms that are not required to allocate all items) is robust to this relaxation of~0, because it only allocates an item to an agent if the agent reports a positive marginal value for the item. Consequently, we shall not bother with this relaxation (that only complicates terminology but has no effect on the results), and assume that undesirable items always have a value of~0.
A natural relaxation for the desirable items is to allow each of them to have an arbitrary value in the range $[1, 1 + \epsilon]$, for some small $\epsilon > 0$ (where the case $\epsilon = 0$ corresponds to dichotomous valuations). Consequently, an agent can attribute slightly different values to her desirable items, {and can \emph{strictly} prefer one desirable item over another}.
We require $\epsilon$ to be sufficiently small
so that the preference order that an agent has over sets of desirable items remains in favor of the larger set. We refer to this setting as that of {\em $\epsilon$-leveled} valuations\footnote{The term \emph{leveled} was introduced in~\cite{BabaioffNT2019} to refer to valuations in which a larger set is always preferred to a smaller one, and we adopt this term here.}. More generally, we have the notion of {\em $\epsilon$-dichotomous} valuations that can be applied also to submodular valuations (and not just additive ones), though we {defer the formal definition to Section~\ref{sec:model}.} 

Let $\rho = \frac{1}{1 + \epsilon}$. If the agents truthfully report their $\epsilon$-dichotomous valuations, one may round the value of each set down to the nearest integer, and by this obtain dichotomous valuations. Thereafter, one may allow the PE mechanism to choose an allocation, and obtain all the guarantees of the PE mechanism up to a multiplicative factor of $\rho$. Likewise, no agent can gain more than a $1 + \epsilon$ factor in her utility by providing an incorrect report to this mechanism (the agent cannot increase the number of desirable items that she receives, but she might possibly be able to manipulate the identity of these items), and hence the mechanism may be referred to as being {\em $\epsilon$-truthful}. However, strictly speaking, this mechanism is not truthful (for $\epsilon$-dichotomous valuations), because an agent who cares about small differences in her utility may indeed find it beneficial to misreport her valuation function. Moreover, once agents misreport their valuation functions, the welfare generated by the resulting allocation might be much smaller than the maximum welfare.

It turns out that there are fundamental limits on truthful mechanisms for $\epsilon$-dichotomous valuations. For the case of submodular (in fact, even just {\em unit demand}) $\epsilon$-dichotomous valuations we show that no truthful allocation mechanism (neither deterministic nor randomized) can approximate the maximum welfare with a ratio better than $\frac{1}{2}$. When restricting attention to $\epsilon$-leveled valuations, there are deterministic truthful mechanisms that generate at least a $\rho$-fraction of the maximum welfare (e.g., let every agent in order of priority select all items that she desires among remaining items). However, we show that for $\epsilon$-leveled valuations, there is no deterministic truthful mechanism that allocates all desirable items and satisfies the following weak fairness requirement: for settings with $n$ agents, if an agent reports $n$ items as desirable, the agent receives at least one of the reported items.  Moreover, there is no truthful allocation mechanism (neither deterministic nor randomized) that maximizes welfare when valuations are $\epsilon$-leveled.

The above impossibility results lead us to consider randomized allocation mechanisms for additive $\epsilon$-dichotomous valuations.
We require our randomized mechanisms to be {\em truthful in expectation} (TIE): misreporting a valuation function cannot increase the expected utility of a agent. TIE is a property that holds before the random allocation mechanism tosses its coins. In addition, our mechanisms preserve the $\epsilon$-truthfulness property mentioned above, and this property holds ex-post (even after the agent sees the coin tosses of the mechanism). As to economic efficiency and fairness properties, we relax them, being content with a $\frac{1}{1 + \epsilon}$ approximations of them. 
Importantly, we require these guarantees to hold in an {\em ex-post} manner, namely, for every realization of the randomness of the underlying randomized allocation mechanism.

Our {second} main result concerns a new randomized allocation mechanism that we refer to as $\meche$.

\begin{theorem}
	\label{thm:leveled-main-intro}
	Let $\epsilon<\frac{1}{nm^3}$ and $\rho = \frac{1}{1 + \epsilon}$. When all $n$ agents have $\epsilon$-leveled valuations,
	the randomized allocation mechanism $\meche$ has the following properties:
	\begin{enumerate}
		\item \label{leveled-truthful}
		{$\meche$ is truthful in expectation and ex-post $\epsilon$-truthful}.  
		\item \label{leveled-welfare}  If all agents are truthful, then the allocation output by $\meche$ provides at least a $\rho$-fraction of the maximum welfare (ex-post).
		\item \label{leveled-fair}  $\meche$ guarantees every truthful agent at least a $\rho$-fraction of her maximin share, and is $\rho$-EF1 (envy free up to one good, up to a multiplicative factor of $\rho$). These guarantees hold ex-post. Moreover, the expected utility received by a truthful agent is at least a $\frac{1}{n}$-fraction of her value for the grand bundle of all items (i.e., proportional in expectation). 
		\item \label{leveled-poly} The mechanism runs in polynomial time.	
	\end{enumerate}
\end{theorem}


Our randomized mechanism $\meche$ is based on the PE mechanism. It first rounds down all reported values to the nearest integer, thus obtaining dichotomous additive valuations.
As simply running the PE mechanism on the rounded valuation does not create a truthful mechanism, we modify the mechanism to obtain truthfulness using some randomization. This randomization involves two components. One, that is very natural from a fairness perspective, is to choose a priority order $\sigma$ uniformly at random (unlike the deterministic PE mechanism for which $\sigma$ is fixed in advance). The other component, a trick that we introduce and that may be of value also elsewhere, is to hold out at random either one or two of the items. For the items not held out, referred to here as the {\em main items}, $\meche$ allocates them using the PE mechanism with priority order $\sigma$. As to the items held out, $\meche$ allocates them using a priority based mechanism, but with a priority order $\sigma'$ that is the \emph{reverse} of $\sigma$. For the first item held out, among the agents that desire it (if there is any), the agent with highest priority (according to $\sigma'$) receives it, and her priority is reduced to being last. If there is also a second item that is held out, then it is allocated according to this new priority order.


Let us briefly explain the main argument why mechanism $\meche$  is truthful in expectation. Consider agent $v$ for which $D_v$ is the set of all {\em desirable} items of non-zero value (hence of value in the range $[1, 1 + \epsilon]$). Truthfulness of the deterministic mechanism PE (for dichotomous valuations) implies that for every outcome of the random coin tosses of $\meche$, reporting her true $D_v$ maximizes the number of desirable items that $i$ receives. Hence the most that $v$ can gain by a non-truthful report is an added value of $\epsilon|D_v| \le \epsilon m \le \frac{1}{nm^2}$. Hence even if there is probability of only $\frac{1}{nm^2}$ of losing a desirable item {by misreporting}, non-truthful reporting {becomes} inferior to truthful reporting. And indeed, the allocation rule for the held-out items is designed such that non-truthful reporting causes a loss of a desirable item with high enough probability, making such a report dominated.
We remark that the proof of truthfulness does not require $\sigma'$ to be the reverse of $\sigma$. The fact that one priority order is the reverse of the other is only used in establishing fairness properties of $\meche$ (item~\ref{leveled-fair} in Theorem~\ref{thm:leveled-main-intro}).

More details concerning our results appear in subsequent sections. Due to space limitations, most proofs (including the statements of some lemmas) are deferred to the appendix.

\subsection{Related Work}
\subsubsection{Previous Work}
{\bf Dichotomous preferences:}
The study of dichotomous preferences was initiated by 
Bogomolnaia and Moulin~\cite{BM2004}. They consider dichotomous matching problems (two-sided unit-demand preferences) and suggest the randomized Lorenz mechanism to get a probabilistic allocation that is fair in expectation.
Within the setting of one-sided markets, the paper of \cite{BM2004} addresses randomized mechanisms for unit-demand valuations. 
We consider the more general class of submodular valuations, and our main focus is on ex-post fairness. 
Dichotomous preferences have been further studied extensively in the literature
for mechanisms without money
\cite{BMS2005, Freitas2010, BouveretL08, KPS2018, Ortega2020}, auction design (with private value scaling) \cite{BLP2009,MishraR2013} and exchanges \cite{RuthSU2005, Aziz2020}.

Maybe the most closely related to our paper is the work of Ortega~\cite{Ortega2020} 
which studies the Multi-unit assignment problem (MAP) with dichotomous valuations.
MAP is a sub-class of the submodular class that slightly extends additive (but does not contain unit demand, for example).
The paper suggests picking a fractional "welfarist" solution (vector of fractional utilities) that is Lorenz dominating among those that maximize welfare. Being fractional, this corresponds to a randomized allocation mechanism rather than a deterministic one. Consequently, the notion of truthfulness used is that of being truthful in expectation. Moreover, the notion of truthfulness is further restricted there, and only allows to conceal desired items in the report, but not to report undesirable items as desired. Under this notion, the solution is strongly group strategyproof. In contrast, the larger class of submodular dichotomous valuations considered in our work contains unit-demand dichotomous valuations, for which no Pareto optimal deterministic allocation mechanism  is  
{strongly} group strategyproof~\cite{BM2004}. 
Being Lorenz dominating, the fractional solution enjoys multiple fairness properties. The work of~\cite{Ortega2020} does not explicitly address the question of to what extent these fairness properties are preserved ex-post, after the fractional solution is rounded to an integer solution.

{\bf Fairness:} The literature of fairness is too extensive to survey in this paper, so we only mention the most related papers. For a general introduction see \cite{BCELP2016,brams1996fair,moulin2004fair}.

Three types of fairness criteria are commonly studied: 

\textbf{Maximin:} 
Budish~\cite{Budish11} has introduced the notion of maximin fairness. 
\citet{KPW18} showed that even for the simple case of additive valuations, the maximin share cannot be given to everyone simultaneously.
The valuations used in the proof are $\epsilon$-leveled, implying  that we cannot aim for exact maximin fairness even for agents with $\epsilon$-leveled valuations.
Constant approximations to the maximin share were presented in \cite{Ghodsi+2018,GM2019}.

\textbf{EF1:} Envy-free up to one good (EF1) was defined by \citet{Budish11}.
EF1 allocations always exist and can be computed efficiently \cite{LiptonMMS04}. 
\citet{BLMS2019} studied the price of fairness for indivisible goods in terms and welfare, showing that welfare loss might be as large as $\Theta(n)$ if valuations are not restricted.
\citet{caragiannis2019unreasonable} proved that for positive additive utilities, a rule based on maximizing Nash social welfare finds an allocation that is both EF1 and Pareto optimal, yet, \citet{lee2017apx} showed that finding such  a maximal NSW allocation is APX-hard.
\citet{BKV18} developed a pseudo-polynomial time algorithm for finding allocations that are EF1 and Pareto efficient. When the valuations are bounded the algorithm runs in polynomial time.

\textbf{EFX:} Envy-free up to any good (EFX) was introduced by \citet{caragiannis2019unreasonable}.
\citet{plaut2018almost} showed existence of EFX allocations when there are two agents. They also exhibited an instance with two agents and items with zero marginals,  where no allocation is both Pareto optimal and EFX. 
Recent papers \cite{caragiannis2019envy,chaudhury2020little} showed that for additive agents there exists an EFX allocation that allocates almost all the items. \citet{amanatidis2020maximum} showed that when items have only two possible values then {NSW maximization} 
implies EFX, and therefore EFX allocations always exist.

{\bf Truthful Fair Allocation Mechanisms:} 
\citet{ABCM2017} characterized deterministic truthful allocation mechanisms for the case of two additive agents (with unrestricted values), implying strong fairness impossibilities.
\citet{ABM2017} studied deterministic truthful allocation mechanisms for approximating the maximin share for additive valuations.
Several papers \cite{mechanism,nguyen2012allocation} have presented  randomized truthful allocation mechanisms that are fair in expectation.
\citet{segal2019fair} studied truthful allocation mechanisms where items can be shared (fractionally allocated) between agents, and showed that the number of shared items can be made smaller than the number of agents.

\cout{
\textbf{More Truthful Fair Mechanisms without money:}\tecomment{I think these less fit, since they are basically on divisible items}
\begin{itemize}
\item \cite{MosselT10} give truthful in expectation randomized division mechanisms that give each player value that is at least $1/n$ of the total. they do it for divisible goods and then say for indivisible goods, when they are bounded and number of goods goes to infinity it converges to 1/k. \tecomment{no money but continues valuations and with every item worth at most a small $\epsilon$, not sure that this one is that interesting}
\item \cite{CGG2013} study mechanism design for proportional fairness with divisible goods and present a mechanism without payment that gives constant approximation. We have indivisible good (and for which constant approximation to proportional fairness is clearly impossible). (no money, divisible goods)
\end{itemize}
}

\textbf{Best-of-Both-Worlds:} 
\citet{freeman2020best}	presented
a recursive
probabilistic serial allocation mechanism for additive valuations.
They showed that ex-ante envy-freeness can be achieved in combination with EF1 ex-post. Moreover, 
they showed that achieving  EF ex-ante, and EF1 and PO ex-post is impossible. 
We, in contrast, are able to achieve all these properties (even for submodular valuations) as we consider dichotomous valuations.
\citet{aleksandrov2015online} considered allocation mechanism of additive dichotomous agents when items arrive online that is both EF ex-ante and EF1 ex-post. {We consider the offline setting, but for the more general submodular valuations case, and get stronger fairness guarantees (EFX, Lorentz domination). } 

\subsubsection{Independent and Concurrent Work}

We devise a deterministic allocation mechanism that is truthful, efficient, Lorenz-dominating and EFX fair.
We also show that a randomized variant of this allocation mechanism is universally truthful and  stochastically envy free, while being efficient and ex-post EFX fair for dichotomous submodular valuations. 
Recently there has been a surge in papers closely related to ours, and we next survey several recent works which present independent and concurrent research that is highly related to our paper.


Concurrent and independent of our work,
\citet{halpern2020fair}
devise an allocation mechanism for the class of \emph{additive} dichotomous valuations.
They show that their $\mbox{MNW}^{\mbox{tie}}$ deterministic allocation mechanism is EF1, PO, and weakly group strategyproof. 
The additive dichotomous valuations setting is a special case of our more general setting of submodular dichotomous valuations, and our PE mechanism and $\mbox{MNW}^{\mbox{tie}}$ are identical for this special case.
\citet{halpern2020fair} also obtain a ``best of both worlds" type result. They consider a randomized allocation mechanism based on rounding the fractional Nash Social Welfare maximizing allocation, and show that their mechanism is ex-ante weakly group strategyproof and ex-post PO and EF1. \citet{aziz2020simultaneously} reproves a similar result, using the same fractional allocation (and noting that it is in fact ex-ante strongly group strategy proof), but using a different rounding procedure. 
Our best-of-both-world result\footnote{The version of our paper posted in February 2020 did not contain Theorem~\ref{thm:RPE-main}, which was added to later versions in response to comments received on the previous version by Herve Moulin.  Likewise, it did not contain Appendix~\ref{sec:group-SP}, whose addition to later versions was motivated by the appearance of \cite{halpern2020fair} and \cite{aziz2020simultaneously}. The February 2020 version did contain the deterministic allocation mechanism for submodular dichotomous valuations, and the randomized allocation mechanism for $\epsilon$-dichotomous valuations (which can be viewed as a best-of-both-world result).}
 (Theorem~\ref{thm:RPE-main}) holds for submodular dichotomous valuations, whereas the results of \cite{halpern2020fair} and \cite{aziz2020simultaneously} hold only for dichotomous additive valuations. Moreover, our randomized mechanism is {\em universally} truthful (agents have no regret even {\em after} they see the realized allocation), whereas it is not known whether this 
property holds for the mechanisms of \cite{halpern2020fair} and \cite{aziz2020simultaneously}. For submodular dichotomous valuations, no mechanism can be simultaneously ex-ante strongly group strategy proof and ex-post EF1 (see Appendix \ref{sec:group-SP}), and in our work we do not attempt to achieve ex-ante weak strategy proofness. 





Another concurrent and independent work is of 
\citet{benabbou2020finding}, that shows how to find  
welfare-maximizing and EF1 allocations for dichotomous submodular valuations in a computational efficient way.
Their result is purely algorithmic and does not consider incentives, whereas  we solve the harder problem of designing a truthful mechanism that obtains all desired fairness and economic efficiency properties (while being computationally efficient).

\section{Model and Preliminaries} \label{sec:model}
We consider settings with a finite set $M$ of $m$ indivisible and non-identical items.
There is a set of $n \ge 2$ agents, denoted by  $V=[n]$, with each agent $v\in V$ having a \emph{valuation function} $f_v$ over sets of items.
The \emph{value} (or utility) of agent $v$ for a set $S\subseteq M$ is denoted by $f_v(S)$.
We always assume that a valuation function $f$ is normalized ($f(\emptyset) = 0$) and non-decreasing ($f(S)\leq f(T)$ for $S\subseteq T \subseteq M$).


\subsection{Valuations}
In this paper we consider several classes of valuation functions.
The \emph{marginal value} of item $a\in M$ given a set $S\subseteq M$ is defined to be $f(a|S)= f(S \cup \{a\}) - f(S) $.
Next we define some properties of valuation functions we will be using:
\begin{itemize}
	\item A valuation function $f$ is  \emph{dichotomous} if the marginal value of any item is either $0$ or $1$, that is, $f(a|S)= f(S \cup \{a\}) - f(S) \in \{0,1\}$ for every set $S\subseteq M$ and item $a\in M$.
	\item A valuation function is \emph{additive} if $f(S)+f(T)=f(S\cup T)$ for any disjoint sets $S,T\subseteq M$.	
	\item A valuation function $f$ is \emph{submodular} if $f(S \cup \{a\}) - f(S) \ge f(T \cup \{a\}) - f(T)$ for every pair of sets $S \subseteq T\subseteq  M$ and every item $a\in M$. 
	\item A valuation function $f$ is a \emph{Matroid Rank Function (MRF)} if there exists a matroid\footnote{{A \emph{matroid} $(\groundset,\indsets)$ is
		constructed from a non-empty ground set $\groundset$ and a nonempty family $\mathcal{I}$ of subsets of $\groundset$, called the \emph{independent}
		subsets of $\groundset$. $\indsets$ must be downward-closed (if $T \in \indsets$ and $S\subseteq T$, then $S \in \indsets$) and satisfy the	\emph{exchange property} (if $S,T\in \indsets$ and $|S|<|T|$, then there is some element $x \in T\setminus S$ such that  $S\cup \{x\}	\in \indsets$).
		The \emph{rank} of a set $S$ is the size of the largest independent set contained in $S$.
	} }
	 for which for every set $S$ it holds that  $f(S)$ is the rank of set $S$ in the matroid.
	\item For $\epsilon\geq 0$, a submodular valuation function $f$ is \emph{$\epsilon$-dichotomous} if the marginal value of any item is either $0$ or belongs to the set $[1,1+\epsilon]$, that is, $f(S \cup \{a\}) - f(S) \in \{0,[1,1+\epsilon]\}$ for every set $S\subseteq M$ and item $a\in M$.
	\item For $\epsilon\geq 0$, an additive valuation function $f$ that is $\epsilon$-dichotomous is called $\epsilon$-leveled.
\end{itemize}

A valuation function $f$ is {\em submodular dichotomous} if it is both submodular and dichotomous. It is easy to see that an MRF is submodular and dichotomous. For the converse direction, given a submodular dichotomous function $f$, consider the family $\mathcal{I}$ that contains those sets $S$ for which $f(S) = |S|$. This family is downward closed, because $f$ is dichotomous. Submodularity of $f$ implies that if $f(T) > f(S)$ there is an item $x \in (T \setminus S)$ for which $f(S \cup \{x\}) = f(S) + 1$. This in turn implies that $\mathcal{I}$ satisfies the exchange property. Hence $\mathcal{I}$ defines a matroid, and $f$ can be seen to be the rank function of this matroid.
Thus, for
brevity we will often refer to a submodular dichotomous valuation function as an MRF valuation.
An interesting special case of submodular dichotomous valuations are such valuations that are additive.
A valuation function $f$ is {\em additive dichotomous} if it is both additive and dichotomous.
Note that a $0$-leveled valuation is simply an additive dichotomous valuation.

For agent $v$ with a submodular valuation $f_v$, 
the set $D_v \subseteq M$ of \emph{items demanded by $v$} includes every item $a\in M$ such that $f_v(\{a\})>0$. 
Observe that for an additive valuation, the value $f_v(S)$ of the set $S$ is $f_v(S)= \sum_{a\in S} f_v(\{a\}) =\sum_{a\in S\cap D_v} f_v(\{a\})$.
If $f_v$ is both dichotomous and additive then $f_v(S)$ is simply $f_v(S) = |S \cap D_v|$ and we call the set $D_v$ \emph{the demand of $v$} (and due to additivity  the value of every item in the demand is independent of other items the agent receives). For item $a\in D_v$ we say that agent $v$ \emph{desires} (or demands, or wants) item $a$.

\subsection{Allocations}
\label{sec:allocations}

We consider mechanisms to allocate items in $M$ to the agents. As we assume that utilities cannot be transfered and there is no money, a mechanism will only specify the allocation function, mapping valuation functions to allocations.  We will mostly consider deterministic allocation functions.

An \emph{allocation} $(A_1,A_2,\ldots, A_n)$ with $A_v\subseteq M $ for every $v\in V$ and $\cup_v A_v\subseteq M$, is an  
assignment of items to agents, possibly leaving some items unallocated.
We denote by $A_v$ the set of items allocated to agent $v$ under allocation $A$. 
The \emph{value} (or \emph{utility}) of allocation $A$ for agent $v$ that has valuation function $f_v$ is $f_v(A_v)$. 

Fix some valuation functions $f=(f_1,f_2,\ldots, f_n)$.
The \emph{welfare of an allocation} $A$ given $f$ is $\sum_v f_v(A_v)$ and an allocation is \emph{welfare maximizing} if there is no other allocation with larger welfare. Note that a welfare maximizing allocation is Pareto optimal.
An allocation $A$ is called \emph{non-redundant} for $f$ if it does not give any agent an item for which she has no marginal value, that is,  for any agent $v$ and any item $a\in A_v$ it holds that $f_v(A_v)>f_v(A_v\setminus\{a\})$, {or equivalently, every strict subset of $A_v$ has strictly lower value for $v$. We note that for MRF valuation $f$, for any non-redundant set $S$ it holds that $f(S)=|S|$}.
{A non-redundant allocation has \emph{maximal size} with respect to $f$, if there is no other non-redundant allocation with respect to $f$ that allocates more items.
We say that an allocation is \emph{reasonable} for $f$ if it both non-redundant and has maximal size with respect to $f$.}
Note that if all agents have dichotomous additive valuations, 
any reasonable allocation is welfare maximizing. Additionally, if agents have $\epsilon$-dichotomous valuations, any reasonable allocation gets at least $\epsap$-fraction of the maximum welfare allocation (Observation \ref{obs:welfare_approx}).

\subsection{Mechanisms}
An allocation mechanism (without money) maps profiles of valuations to an allocation. That is, given valuation functions $f = (f_1, \ldots, f_n)$ an \emph{allocation mechanism $\mech$} outputs an allocation $A= \mech(f) = \mech(f_1, \ldots, f_n)$. We sometimes abbreviate and call an allocation mechanism simply a \emph{mechanism}.
A mechanism asks each agent to report a valuation function, getting a report $f'_v$ from each agent $v$, and allocates the items by running the mechanism on the reported valuations $(f'_1, \ldots, f'_n)$, that is, mechanism $\mech$ outputs   $A= \mech(f'_1, \ldots, f'_n)$. We are interested in mechanisms that are truthful, that is, give agents incentives to report their valuation function truthfully. A mechanism $\mech$ is \emph{truthful} if for every agent $v$, reporting $f_v$ is a weakly dominant strategy (maximizes her value given any reports of the other agents).

We say that a mechanism $\mech$ has property $P$ if for any input $f$, its output allocation $A=\mech(f)$ has property $P$.
For example, a mechanism is \emph{reasonable} if for any $f$ the allocation $A=\mech(f)$ is reasonable for $f$.

\subsection{Fairness}

The list below presents standard fairness conditions that one may desire.

\begin{enumerate}
	
	\item The {\em maximin share} of an agent $i$ with valuation $f_i$, denoted by $\mbox{maximin}(f_i)$,  is the maximum over all partitions of the items into $n$ disjoint bundles $S_1, \ldots, S_n$ of the minimum value according to $f_i$ of a bundle, $\min_{j\in [n]} f_i(S_j)$. The optimal partition depends on $f_i$. An allocation that gives each agent a bundle of value at least as high as his maximin share is called {\em maximin fair}.
	
	\item An allocation is {\em envy free} (EF) if every agent prefers the bundle that he himself received over every bundle that some other agent received. Formally,  for every $i \in [n]$, let $f_i$ denote the valuation function of agent $i$ and let $A_i$ denote the bundle received by agent $i$. Then an allocation is envy free if for all $i,j \in [n]$ it holds that $f_i(A_i) \ge f_i(A_j)$. In the context of allocation of indivisible goods, envy freeness is incompatible with economic efficiency. For example, if there is only one item and two agents who desire it, the only envy free solution is not to allocate the item at all. Consequently, the literature considers the following relaxed versions of envy freeness:
	
	\begin{enumerate}

		\item {\em Envy free up to one good} (EF1). The envy free condition is relaxed as follows: for all $i,j \in [n]$ either $f_i(A_i) \ge f_i(A_j)$, or \emph{there is an item} $e \in A_j$ such that $f_i(A_i) \ge f_i(A_j \setminus \{e\})$.
		
		\item {\em Envy free up to any good} (EFX). For all $i,j \in [n]$ either $f_i(A_i) \ge f_i(A_j)$, or \emph{for every item} $e \in A_j$ it holds that $f_i(A_i) \ge f_i(A_j \setminus \{e\})$. EFX is a stronger property than EF1.
		
	\end{enumerate}
	\item Given an allocation $A = (A_1, \ldots, A_n)$ and valuation functions $f = \{f_1, \ldots, f_n\}$, the {\em utility vector} is $u_{A,f} = (f_1(A_1), \ldots, f_n(A_n))$, and the sorted utility vector $s_{A,f}$ is a vector whose entries are those of $u_{A,f} $
	sorted from smallest to largest. We impose a {\em lex-min} order among sorted vectors, where $s_1 >_{lexmin} s_2$ if there is some $k \in [n]$ such that $s_1(k) > s_2(k)$ and for every $1 \le j < k$ we have that $s_1(j) = s_2(j)$. Given the valuation functions $f$, an allocation $A$ is maximal in the lex-min order if for every other allocation $A'$ we have that $s_{A,f} \ge_{lexmin} s_{A',f}$. We refer to such an allocation as a {\em lex-min allocation}. Given $f$, a lex-min allocation always exists (as the set of allocations is finite).
	
	\item Using notation as above, we also impose a {\em Lorenz domination} partial order over sorted vectors, where $s_1 \ge_{Lorenz} s_2$ if for every $k \in [n]$, the sum of first $k$ entries in $s_1$ is at least as large as the sum of first $k$ entries in $s_2$. A {\em Lorenz dominating allocation} is an allocation that Lorenz dominates every other allocation. Given the valuation functions $f$, a Lorenz dominating allocation need not exist, but if it does exist, then it is also a lex-min allocation.
	
	\item Given valuation functions $f$, an allocation $A$ will be referred to as {\em min-square} if it maximizes welfare ($\sum_i f_i(A_i)$), and conditioned on maximizing welfare, it minimizes $\sum_i (f_i(A_i))^2$. A min-square allocation always exists. It is an allocation that minimizes the variance of utilities among agents, conditioned on maximizing the welfare.
	
	\item Given valuation functions $f$, an allocation $A$ is said to maximize the {\em Nash Social Welfare} (NSW) if it maximizes the product $\prod_i f_i(A_i)$. (Formally, such an allocation maximizes NSW relative to the {\em disagreement point} of not allocating any item.) Given $f$, a maximum NSW allocation always exist, though it need not maximize welfare.
	
\end{enumerate}
Each of the notions of maximin share,  EF1 and EFX can be relaxed to hold only up to an multiplicative term of {$\alpha\in [0,1]$}, and we use the notation $\alpha$-maximin, $\alpha$-EF1 and $\alpha$-EFX to denote these approximate fairness notions in which an agent gets at least $\alpha$-fraction of the appropriate share (see Section \ref{app:approx-fair}  in the appendix for formal definitions).

\section{Submodular Dichotomous Valuations}
\label{sec:submodular}
\cout{
Recall that there are $m$ items that need to be allocated to $n$ players.
In this section we consider a class of valuation functions that combines the dichotomous requirement with submodularity. A function $f$ is {\em submodular dichotomous} if the following properties hold:

\begin{itemize}

\item $f(\emptyset) = 0$.

\item $f(S \cup \{a\}) - f(S) \in \{0,1\}$ for every set $S$ and item $a$.

\item $f(S \cup \{a\}) - f(S) \ge f(T \cup \{a\}) - f(T)$ for every pair of sets $S \subset T$ and every item $a$.

\end{itemize}
}
In this section we consider valuation functions that are both dichotomous and submodular.
It is known (follows from the matroid exchange property)
that a function is submodular dichotomous if and only if it is a rank function of a matroid. Consequently, for brevity, we shall refer to submodular dichotomous valuations as MRFs (Matroid Rank Functions).


	In this section we prove our first main result:
\begin{theorem}
	\label{thm:MRF-main}
	There exists an allocation mechanism for agents with MRF valuations with the following properties:
	
	\begin{enumerate}
		\item \label{item:MRF-truthful} For agents with MRF valuations, being truthful is a dominant strategy.
		
		\item \label{item:MRF-welfare} When agents are truthful the allocation is welfare maximizing.
		
		\item \label{item:MRF-fair}  When agents are truthful, the allocation of the mechanism is a Lorenz dominating allocation, it is lex-min, it is min-square, and maximizes NSW. It is  also EFX (and hence also EF1). If furthermore, the valuations are additive dichotomous, the allocation is also maximin fair.
		
		\item \label{item:MRF-poly}  If the agents have MRF valuations with succinct representations, then the mechanism can be implemented in polynomial time.
		
		
	\end{enumerate}
	
\end{theorem}


One aspect of the proof of the theorem involves showing that a Lorenz dominating allocation exists. Lorenz domination can be shown to imply the desired welfare and fairness properties. However, simply picking an arbitrary Lorenz dominating allocation does not guarantee truthfulness (see Example~\ref{ex:Lorenz}). Hence a major part of the proof of the theorem is to show that a particular choice of a Lorenz dominating solution does ensure truthfulness.

\begin{example}
\label{ex:Lorenz}

Consider a setting with two items, three agents ($p_1$, $p_2$ and $p_3$), additive dichotomous valuations, and the following allocation mechanism $M$. Every agent is asked to report the set of items that she desires. Then, $M$ selects an allocation based on the following principles. One principle is that every reported item is allocated to some agent that reports it as desired. Another principle is that no agent gets both items, if there is another agent who wants at least one of the items. The combination of these two principles ensures that $M$ only picks Lorenz dominating allocations (if it knows the true valuation functions of the agents). In addition, $M$ has two tie breaking rules. The first rule is that agents who report just one item have priority over those who report two items. (Agents that report no item get no item and hence are ignored.) The second rule is that among those agents that report the same number of items, agents of lower index have higher priority than agents of higher index. Hence for example, if $p_1$ and $p_2$ report both items and $p_3$ reports only one item, then $p_3$ gets the item that she reports (by the first tie breaking rule), and $p_1$ gets the remaining item (by the second tie breaking rule). As another example, if all agents report both items, than $p_1$ and $p_2$ each get one item (by the second tie breaking rule). The combination of the two examples demonstrates that even though $M$ only picks Lorenz dominating allocations, it is not truthful. That is, if all agents desire all items and agents $p_1$ and $p_2$ report truthfully, then $p_3$ gains an item by falsely reporting that she desires only one item.
\end{example}

\subsection{Lorenz Dominating Allocations}\label{sec:Lorntz-exists}
The following proposition puts together several observations {regarding fairness properties of Lorenz dominating allocations}, most (if not all) of which are known.

\begin{restatable}{proposition}{LorenzIsGood}
\label{pro:LorenzIsGood}
Given any (normalized and monotone) valuation functions $f = (f_1, \ldots, f_n)$, a Lorenz dominating allocation, if it exists, also maximizes welfare, is lex-min, is min-square, and maximizes NSW. If moreover the valuation functions are MRFs, a Lorenz dominating allocation that is non-redundant (see definition of non-redundant in Section~\ref{sec:allocations}) is also EFX (and hence also EF1). If furthermore, the valuations are additive dichotomous, a Lorenz dominating allocation is also maximin fair.
\end{restatable}

For MRF valuations Lorenz dominating allocations might not be maximin fair, but are approximately so.
\begin{restatable}{proposition}{LorenzIsBad}
\label{pro:LorenzIsBad}
There are MRF valuation functions with respect to which no Lorenz dominating allocation is maximin fair. For every collection of MRF valuation functions, in every Lorenz dominating allocation every agent gets at least half her maximin share.
\end{restatable}


In view of Propositions~\ref{pro:LorenzIsGood} and~\ref{pro:LorenzIsBad}, we choose Lorenz domination as our fairness requirement. As we shall see in Theorem~\ref{thm:LorenzDuttaRay}, in our setting of MRF valuations, a Lorenz dominating allocation always exists, and often, more than one such allocation exists. For example, if there is only one item and all agents desire it, then allocating the item to any of the agents is a Lorenz dominating allocation. We now wish to address truthfulness of the allocation mechanism. This will be achieved by implementing a particular choice among Lorenz dominating allocations. This choice will be guided by two principles.

The first principle is that the allocation will be {\em non-redundant}. Namely, for every agent, the allocation is such that the set of items given to the agent does not contain redundant items that give the agent no marginal value. In our setting, this is equivalent to requiring that the set of items received by an agent forms an independent set in the matroid underlying the MRF of the agent.

The second principle is that of imposing some arbitrary {\em priority} order among the agents, fixed independently of their valuations. Among the possibly many Lorenz-dominating allocations that may exist, we choose one that favors the higher priority agents as much as possible. Still, there may be several different allocations that satisfy this condition, but any two of them will be equivalent in terms of the utilities that the agents (who have MRF valuation functions) derive from them.

W.l.o.g., let the priority order be such that agent $i$ has priority $i$ (agent~1 has highest priority, agent~$n$ has lowest priority). A convenient mathematical way to reason about the priority order is as follows. Add to the instance $n$ auxiliary items $a_1, \ldots ,a_n$. For every agent $i \in [n]$, pretend that the marginal value of item $a_i$ is $\frac{i}{n^2}$ to agent $i$ (regardless of any other items that agent $i$ may hold), and the marginal value of $a_j$ with $j \not= i$ is~0. With the auxiliary items, the new valuations $f' = (f'_1, \ldots, f'_n)$ of agents satisfy $f'_i(S) = f_i(S \cap M) + \frac{i}{n^2}|S \cap \{a_i\}|$. They are not MRFs (because the marginals of auxiliary items are not in $\{0,1\}$), but they are still gross substitutes (because each $f'_i$ is a sum of a gross substitute function $f_i$ on the original items and a gross substitute function on the auxiliary items). In every welfare maximizing allocation, for every $i\in [n]$, item $a_i$ is given to agent $i$. Given the auxiliary item, when allocating the original items, a Lorenz dominating allocation will break ties in favor of higher priority agents, as they derive less value from the auxiliary items.

\begin{restatable}{theorem}{LorenzDuttaRay}
\label{thm:LorenzDuttaRay}
Given MRF valuations $f = (f_1, \ldots, f_n)$ and the auxiliary items (giving rise to new valuations $f' = (f'_1, \ldots, f'_n)$), there is a Lorenz dominating allocation $A'$. Moreover there is a unique vector of utilities (the vector $u_{A',f'} = (f'_1(A'_1), \ldots, f'_n(A'_n))$ shared by all Lorenz dominating allocations.
Removing the auxiliary items from the Lorenz dominating allocation results in an allocation $A$ that is Lorenz dominating with respect to the original MRF valuations.
\end{restatable}

\begin{proof}
Consider the following  function $W$ that we shall refer to as a {\em welfare function}. Given a set $M$ of indivisible items, a set $V$ of $n$ agents, and valuation functions $f_1, \ldots, f_n$, the function $W$ is a set function defined over the agents. Given a set $S \subseteq V$, $W(S)$ is the maximum welfare attainable by the set $S$. Namely, $W(S) = \max_{A = (A_1, \ldots, A_n)}[\sum_{i\in S} f_i(A_i)]$. Lemma~\ref{cor:GStoSubmodular} shows that given our valuation functions $f' = (f'_1, \ldots, f'_n)$ (the MRFs, augmented with the auxiliary items), the respective welfare function $W$ is submodular.
Dutta and Ray~\cite{DR89} prove that if $W$ is submodular, then a Lorenz dominating allocation exists. Consequently, with our valuation functions $f'$, a Lorenz dominating allocation $A'$ exists. (At this point we are only concerned with the existence of a Lorenz dominating allocation. Algorithmic aspects are deferred to the proof of Theorem~\ref{thm:MRF-poly-time}.)
Uniqueness of the vector of utilities is a consequence of the fact that with the auxiliary items, the utility of an agent uniquely identifies the agent. Consequently, any two different vectors of utilities give two different sorted vectors. The sorted vector of a Lorenz dominating allocation is unique (by definition of the Lorenz domination partial order), and hence the (unsorted) vector is also unique.

Removing the auxiliary items from a Lorenz dominating allocation $A'$ gives an allocation $A$ that is Lorenz dominating with respect to the original MRF valuations $f$. For the sake of contradiction, suppose otherwise, that there is some allocation $B$ such that $A$ does not Lorenz dominate $B$. Then there is some $k \le n$ such that the sum of the first $k$ terms of the sorted vector of $B$ is larger than the sum of the first $k$ terms of the sorted vector of $A$. As the values of both sums are integer, the difference between the two sums is at least~1. Consequently, even $A'$ does not Lorenz dominate $B$, because the total contribution of auxiliary items is at most $\sum _{i=1}^n \frac{i}{n^2} < 1$, contradicting the assumption that $A'$ was Lorenz dominating with respect to $f'$.
\end{proof}

\subsection{The Prioritized Egalitarian (PE) Mechanism}\label{sec:PE-mech}
We can now present our allocation mechanism for MRF valuations, that we refer to as the {\em prioritized egalitarian (PE) mechanism}. We assume for this purpose that each MRF $f_i$ has a succinct representation (of size polynomial in the number $m$ of items) such that given this representation, for every $S$ one may compute $f_i(S)$ (answer {\em value queries}) in time polynomial in $m$.

\begin{enumerate}

\item The mechanism imposes an arbitrary priority order $\sigma$ among the agents. For simplicity and without loss of generality, we assume that the order is from~1 to $n$, where agent~1 has highest priority.

\item Every agent is requested to report his MRF to the mechanism. A report that is not an MRF (or failure to provide a report at all) is considered illegal, and is replaced by the MRF that is identically 0 (and consequently, the non-redundant allocation will not give such an agent any item).

\item Given the reported MRF valuation functions $r_1, \ldots, r_n$, the mechanism computes a non-redundant Lorenz dominating allocation $A = (A_1, \ldots, A_n)$ with respect to these reports and $\sigma$ (as implied by Theorem~\ref{thm:LorenzDuttaRay}), and gives each agent $i$ the respective set $A_i$.

\end{enumerate}

We now show that the PE mechanism is truthful. For this we introduce some notation. Given a valuation function $f_i$ and a set $D$ of items, we use $f_{i|D}$ to denote the function $f_i$ restricted to the items of $D$. Namely, for every set $S$, $f_{i|D}(S) = f_i(S \cap D)$. We note that if $f_i$ is an MRF, then so is $f_{i|D}$. Truthfulness will be a consequence of the following properties of the allocation mechanism.

\begin{itemize}

\item We say that an allocation mechanism is {\em faithful} if the following holds for every collection $f = (f_1, \ldots, f_n)$ of valuation functions and for every agent $i$. Let $A_i$ denote the allocation of the mechanism to agent $i$ when the reported valuation functions are $f$. Then if instead agent $i$ reports valuation function $f_{i|A_i}$ (and {the reports of the other agents} 
remain unchanged), then the allocation to agent $i$ remains $A_i$. We say that an allocation mechanism is {\em strongly faithful} if it is faithful, and in addition, for every set $A' \subset A_i$, if agent $i$ reports valuation function $f_{i|A'}$ (and {the reports of the other agents} 
remain unchanged), then the allocation to agent $i$ becomes $A'$.

\item 
We say that an allocation mechanism is {\em monotone} if the following holds for every collection $f = (f_1, \ldots, f_n)$ of valuation functions, every agent $i$, and every two sets of items $S$ and $T$ with $S \subset T$. Let $A_{i|S}$ denote the allocation of the mechanism to agent $i$ when the reported valuation function for agent $i$ is $f_{i|S}$ and the remaining reports are as in $f$. Then if instead agent $i$ reports a legal (see remark that follows) valuation function $f_{i|T}$ (and the remaining reports remain unchanged), then the allocation $A_{i|T}$ to agent $i$ satisfies $f_i(A_{i|T}) \ge f_i(A_{i|S})$. ({\em Remark.} It may happen that $f_i$ is not an MRF, $f_{i|T}$ is not an MRF, but $f_{i|S}$ happens to be an MRF. In this case the PE mechanism might produce a nonempty $A_{i|S}$ and an empty $A_{i|T}$, violating the inequality $f_i(A_{i|T}) \ge f_i(A_{i|S})$. For this reason we do not impose the monotonicity condition if the valuation function $f_{i|T}$ is illegal with respect to the underlying allocation mechanism.)

\end{itemize}

We next prove that the PE mechanism is truthful. Some of the claims used in the proof are presented and proved in the appendix.
\begin{theorem}
The PE mechanism is truthful for agents with MRF valuations. Namely, for every agent with an MRF valuation, reporting her true valuation function maximizes her utility, for any reports of the other agents.
\end{theorem}

\begin{proof}
Consider an arbitrary agent $v$. Fix the reported valuation functions of all other agents. All these reported valuation functions {can be assumed to be} 
MRFs, because the PE mechanism replaces every non-MRF reported function by the all~0 MRF. Let $f_v$ be the MRF valuation of $v$. Let $A_v$ denote the set of items that $v$ receives when reporting $f_v$. Suppose now that instead $v$ reports a different valuation function $f'_v \not= f_v$, and receives an allocation $A'_v$. We need to show that $f_v(A_v) \ge f_v(A'_v)$.

We may assume that $f'_v$ is an MRF, as otherwise $v$ gets no item and $f_v(A_v) \ge 0 = f_v(\emptyset)$. Change the report $f'_v$ to $f'_{v|A'_v}$. By faithfulness (Proposition~\ref{pro:faithful}) the allocation to $v$ remains $A'_v$. Let $B \subseteq A'_v$ be a subset of smallest cardinality for which $f_v(B) = f_v(A'_v)$.
Necessarily, $|B| = f_v(A'_v) = f_v(B)$, and $B$ is a maximum size subset of $A'_v$ that is independent with respect to the matroid underlying $f_v$.
Change the report $f'_{v|A'_v}$ to $f'_{v|B}$. By strong faithfulness (Lemma~\ref{lem:strongfaith}) the allocation to $v$ becomes $B$. Now change the report $f'_{v|B}$ to $f_{v|B}$. This changes nothing because as functions $f'_{v|B} = f_{v|B}$ (both $f'_v$ and $f_v$ give value~1 to items of $B$, value~0 to other items, and are additive over $B$ -- for $f'_v$ additivity follows because the allocation is non-redundant, and for $f_v$ because $B$ was chosen to be an independent set of the matroid), and hence the allocation to $v$ remains $B$. Finally, change the report $f_{v|B}$ to $f_v$. By monotonicity (Lemma~\ref{lem:monotone}), the resulting allocation to $v$ (which is now simply $A_v$) has value to $v$ at least as high as $B$ does. We conclude that $f_v(A_v) \ge f_v(B) = f_v(A'_v)$, as desired.
\end{proof}

Finally, we prove that the PE mechanism can be computed efficiently.

\begin{restatable}{theorem}{MRFpolytime}\label{thm:MRF-poly-time}
If the agents have MRF valuations with succinct representations, then the PE mechanism can be implemented in polynomial time.
\end{restatable}



\subsection{Best-of-Both-Worlds via Random Priorities}
\label{sec:rpe}
We have presented the PE mechanism that is truthful, welfare maximizing and EFX (among other fairness properties). Yet, as this is a deterministic mechanism, it cannot be envy-free (deterministic envy freeness is clearly impossible: consider a single desired item and two agents). 
In this section we show that envy can be eliminated (ex-ante) by running the PE mechanism with uniformly random priorities, and that this holds even if agents are not risk neutral.  

The \emph{random priority egalitarian (RPE) mechanism} is the mechanism that first assigns the agents priorities uniformly at random, and then runs the PE allocation mechanism with the drawn priorities. 
This mechanism is universally truthful (truthful for any realization of the random priorities), welfare maximizing and obtains all the fairness properties of the PE mechanism ex-post. 
Moreover, as we next show, it is also \emph{stochastically envy free}. 
This establishes a best-of-both-worlds result: 
both stochastic envy freeness of the randomized allocation, and EFX (among other fairness properties) ex-post.


For given valuation functions $(v_1,v_2,\ldots,v_n)$, 
a distribution over allocations is {\em stochastically envy free} if for every two agents $i$ and $j$, and for every value $t$,
$$Pr[v_i(A(i)) \ge t] \ge Pr[v_i(A(j)) \ge t],$$
where the probability is taken over the choice of random allocation according to the given distribution.

We note that this notion of stochastic envy-freeness 
implies {\em ex-ante envy-freeness}, that is, it implies $E[v_i(A_i)] \ge E[v_i(A_j)]$, but it is stronger (see Example~\ref{ex:ex-ante-EV-vs-stoc}), and it 
{implies that for any risk attitude, and not only when an agent is risk neutral, he prefers his own lottery over the lottery of any other agent (e.g., he can be risk seeking or risk averse).} 
Additionally, note that any ex-ante envy-free mechanism  that is not wasteful (i.e., for each agent, the marginal value of the set on unallocated items is always zero) is also ex-ante proportional for subadditive valuations (Observation \ref{obs:ex-ante-EF-prop}).


We next present our  result for the RPE mechanism, showing that it obtains the "best-of-both-worlds'': 
it is universally truthful, welfare maximizing and  stochastically envy-free as well as EFX ex-post.
This result does not require agents to be neutral to risk.  
\begin{restatable}{theorem}{RPE}
	\label{thm:RPE-main}
	The \emph{random priority egalitarian (RPE) mechanism} has the following properties when agents have submodular dichotomous valuations:
	
	\begin{enumerate}
		\item Being truthful is a dominant strategy for any realization of the priorities (universally truthful).
		
		\item When agents are truthful the realized allocation is welfare maximizing.
		
		\item  When agents are truthful, the realized allocation of the mechanism is a Lorenz dominating allocation, and consequently it enjoys additional fairness properties, including maximizing Nash social welfare, and being envy-free up to any item (EFX). If furthermore, the valuations are additive dichotomous, the allocation gives every agent at least her maximin {share.} 
	
	 \item The mechanism is stochastically envy-free, and thus is ex-ante envy free as well as ex-ante proportional.
	 
		
		
	\end{enumerate}
	
\end{restatable}
The proof of Theorem~\ref{thm:RPE-main} is deferred to Appendix~\ref{app:RPE}.

\section{$\epsilon$-Leveled Valuations}
\cout{
	List of properties we need of $\mechb$
	\begin{enumerate}
		\item players with higher priority (with respect to $\mechb$)
		do not envy players with lower priority.
	\end{enumerate}
}
	
We have shown (Theorem \ref{thm:MRF-main}) that for submodular dichotomous (MRF) valuations there is a deterministic truthful allocation mechanism that always outputs welfare maximizing allocations that is Lorentz dominating and thus satisfies multiple fairness properties.
In this section we want to explore the robustness of this result, by examining  the extent to which we can relax the assumption that the marginal value of any item with respect to any set is either \emph{exactly} 0 or \emph{exactly} 1.
A natural relaxation of this assumption is to allow marginals to be \emph{almost} those values.
It turns out the result for MRF valuations is not sensitive to undesired items having negative utility instead of 0, so we focus on relaxing the assumption that positive marginals must be $1$, and instead allow the positive marginals to be in $[1,1+\epsilon]$, allowing an agent not to be indifferent between desired items.
This is the case of $\epsilon$-dichotomous valuations in which the marginal value of every item is either 0 or in $[1,1+\epsilon]$.

While our prioritized egalitarian mechanism presented in Section \ref{sec:submodular} for MRF valuations maximizes welfare, it is easy to see that once we allow desired items to have different values, no truthful allocation mechanism will be welfare maximizing. (If the mechanism always maximizes welfare, it is easy to construct examples where agents have incentives to report a value of $1 + \epsilon$ for items that they value at~1.) 
Nevertheless, as valuations are $\epsilon$-dichotomous, reasonable allocations have maximal size and almost maximize welfare (getting at least $\epsap$ fraction), {see Observation \ref{obs:welfare_approx}}, so we aim for a truthful mechanism that always outputs an allocation with welfare at least $\epsap$ fraction of the maximum welfare. We call such an allocation \emph{approximately welfare maximizing}.

Unfortunately,  for  $\epsilon$-dichotomous agents, truthfulness and the requirement to output approximately welfare maximizing 
allocations are at odds, even when allowing randomized mechanisms. 
Indeed, the following example shows that for every $\epsilon >0 $  there is no truthful mechanism (deterministic or randomized\footnote{For randomized mechanism the example is such that a lie stochastically dominates the truth.}) that always returns an approximately welfare maximizing 
allocation for $\epsilon$-dichotomous submodular valuations (that are unit demand, \emph{not} additive).

\begin{example}\label{example:eps-dic-not-truthful}
	Let the set of items be $M =\{L,H\}$, and there are two agents with $\epsilon$-dichotomous valuations.
	Both agents $1,2$ are unit demand with identical valuation function $f$ satisfying $f(L)=1,f(H)=1+\epsilon$ (both agents slightly prefer the high value item $H$ over the low value item $L$).
	When both agents report their true valuation, at least one of them has positive probability of not getting item $H$. W.l.o.g., let agent $1$ be that agent. If agent $1$ reports a valuation $f_1(L)=0,f_1(H)=1+\epsilon$ (that is, that he does not want the low value item at all), and agent $2$ reports truthfully, the only approximately welfare maximizing 
	allocation is the one in which agent $1$ receives item $H$ (and agent $2$ gets $L$) and therefore agent $1$ gets higher utility by lying.
\end{example}

Given Example \ref{example:eps-dic-not-truthful} above it is clear we could not achieve the desired properties for $\epsilon$-dichotomous valuations,
so we consider the problem of designing a truthful and approximately welfare maximizing 
mechanism only for the more restricted class of $\epsilon$-dichotomous \emph{additive} valuations ($\epsilon$-leveled valuations). 
A truthful mechanism that always outputs approximately welfare maximizing 
allocations for $\epsilon$-leveled valuations indeed exists: order the agents arbitrarily  and allow each agent in turn to pick all items he desires out of the remaining items.
Yet, this mechanism is clearly very unfair, as if the first agent desires all items, he will take all items and other agents will get nothing.
Thus, we ask whether there exists a deterministic truthful allocation mechanism for $\epsilon$-leveled agents that always outputs approximately welfare maximizing 
allocations and is fair. 
Unfortunately, we show that even with two agents and three items, there does not exist such a deterministic mechanism that satisfies the following  minimal fairness property: if an agent desires two items, he gets at least one of them. This precludes a mechanism that satisfies any of our fairness properties (e.g., maximin fair, EF1).

\begin{restatable}{proposition}{nodetleveled}\label{prop:no-det-leveled}
	There is no deterministic truthful allocation mechanism for allocating three items to two $\epsilon$-leveled agents
	that always outputs approximately welfare maximizing 
	allocations, and in case that an agent desires two items, he gets at least one of them.
\end{restatable}

We remark that~\citet{ABCM2017} present impossibility results concerning fair allocation mechanisms for two agents with additive valuations, but the setting of Proposition~\ref{prop:no-det-leveled} is not captured by these results, as it restricts agents to be $\epsilon$-leveled, and not just additive.



Given this negative result for deterministic mechanisms, we will consider randomized mechanisms. We aim for a mechanism that is truthful in expectation (over the randomization of the mechanism), is approximately welfare maximizing, 
and is as fair as possible (ex-post, for any realization). Clearly such a mechanism that always outputs a Lorentz dominating allocation does not exist (as such an allocation might not exist\footnote{Consider two items and two agents. The first agent has value $1+\epsilon$ for each item, while the second has value $1$ for each item. A Lorentz dominating allocation, if exists, must be welfare maximizing and must also maximize the NSW (by Proposition \ref{pro:LorenzIsGood}), yet clearly the first requires both items to got to the first agent, while the second require dividing the items between the two.}), yet, maybe a Lorentz dominating allocation exists when rounding down the value of every set to the nearest integer (essentially only counting the number of desired items in the set)?

We first formalize the notation of rounding down a valuation function.
Given a function $f$, we define $\hat{f}$ to be the \textit{floor} of $f$ as follows:
\begin{equation}
\hat{f}(S) =\lfloor f(S) \rfloor\ \ \ \forall S\subseteq M. \label{eq:dichotomous_lower}
\end{equation}

We observe that such a rounding does not change the value of any set by much:
\begin{restatable}{observation}{dichotomous}
	For non-negative $\epsilon< \frac{1}{m}$, and any $\epsilon$-dichotomous valuation $f$, the floor function $\hat{f}$ is dichotomous and for  every set $S \subseteq M$,
	\begin{equation}
	f(S) \leq (1+\epsilon)\hat{f}(S). \label{eq:dichotomous_upper}
	\end{equation}
	Additionally, if $f$ is $\epsilon$-leveled then the function $\hat{f}$ is dichotomous and additive.
	\label{obs:dichotmous}
\end{restatable}

As moving from $\epsilon$-dichotomous valuations to their floors only results with very small changes to the valuations (up to a {multiplicative term of $(1+\epsilon)$,} by Equation (\ref{eq:dichotomous_upper})), one may hope that mechanisms that achieve some fairness properties for dichotomous valuations, also have  these properties holding approximately when they run on the floor valuations. As Lorentz dominating allocations have many desired fairness properties (recall Proposition \ref{pro:LorenzIsGood}) it would be most attractive to show that there exist randomized truthful in expectation allocation mechanism for $\epsilon$-leveled agents that for any valuations $f$ always outputs a Lorentz dominating allocation with respect to $\hat{f}$. Unfortunately, this is not achievable as the next proposition shows.

\begin{restatable}{proposition}{norandLorentzleveled}\label{prop:no-rand-Lorentz-leveled}
	Consider any of the following fairness properties:  being Lorentz dominating, being min-square, maximizing NSW, being lex-min, being EFX. Every randomized truthful in expectation allocation mechanism for allocating two items to two $\epsilon$-leveled agents either sometimes fails to satisfy this property ex-post, even with respect to $\hat{f}$, or sometimes fails to be reasonable (fails to allocate a demanded item to some agent who demands it). 
\end{restatable}

Given this negative result, the fairness properties that we can hope the mechanism could achieve are EF1 and maximin fairness.
Moreover, as the mechanisms we consider are randomized, we can hope to also get proportionality in expectation (clearly it cannot be obtained ex-post).
Recall that $S$ gives agent $v$ his \emph{proportional share} if  $f_v(S) \ge f_v(M)/n$. An allocation $A = (A_1, \ldots, A_n)$  is \emph{proportional} if
	$f_v(A_v)\geq f_v(M)/n$ for every agent $v$, and a randomized  mechanism is \emph{proportional in expectation} if for every agent, the expected value of the allocation to the agent is at least $f_v(M)/n$ when truthful.

In light of the above, we aim to design a randomized mechanism that is truthful in expectation for $\epsilon$-leveled valuations, always outputs approximately welfare maximizing 
allocations, is EF1 and maximin fair with respect to the floor valuations $\hat{f}$, and is proportional in expectation (for $f$).
Note that EF1 and maximin fair with respect to the floor valuations $\hat{f}$
implies that the corresponding inequality holds, up to a small multiplicative loss, with respect to the actual valuations $f$ (e.g., for maximin, the agent is getting at least $\epsap$-fraction of the maximin share).

We emphasize that we are aiming for \emph{exact} truthfulness and we are not satisfied with mechanisms for which truth telling is only approximately best\footnote{By Observation \ref{obs:guarantees} we get that if $\mech$ is truthful for dichotomous agents then $\hat{\mech}$ is $\epsilon$-truthful for $\epsilon$-dichotomous valuations.}. We next present a mechanism that achieves all these properties. 	

The mechanism we design will work on the floor valuations, or equivalently, ask agents for the \emph{set} of items they desire (have value in $[1,1+\epsilon]$),
and not try to elicit the  exact value of each item. This approach allows using mechanisms designed for dichotomous valuations to work for $\epsilon$-dichotomous valuations.
Given a mechanism $\mech$ for dichotomous valuations we denote by $\hat{\mech}$ the mechanism for $\epsilon$-dichotomous valuations such that $\hat{\mech}(f) = \mech(\hat{f}) $ for $\epsilon <\frac{1}{m}$.
It is easy to see that if $\mech$ gives each truthful agent with dichotomous valuations his maximin share, then $\hat{\mech}$ gives each truthful agent with  $\epsilon$-dichotomous valuation at least $\epsap$-fraction of his maximin share 
(and similarly, the inequalities defining  EF1 and EFX approximately hold, see Observation  \ref{obs:guarantees}).


\subsection{A Truthful Mechanism for Leveled Valuations}
In this section, we present a truthful\footnote{We are slightly abusing the notion of truthfulness here, as the agents are not asked to directly report their $\epsilon$-leveled valuations. Rather, each agent is only being asked to report the set of items he demands (which describes an additive dichotomous valuations), and we say that an agent is ``truthful" if he reports that set truthfully.    }
in expectation allocation mechanism for $\epsilon$-leveled valuations (for small enough $\epsilon>0$). Recall that an $\epsilon$-leveled valuation is an additive valuation in which the marginal value of every item is either in $[1,1+\epsilon]$, or is zero\footnote{Actually, the truthfulness result we present can be extended to the case that instead of zero marginals we allow any non-positive marginals (allowing for non-monotone valuations).
}. 
We prove that when agents are truthful, the allocation is reasonable (and thus approximately welfare maximizing) and satisfies several notions of fairness. It is $\epsap$-maximin fair and $\epsap$-EF1
	(both are ex-post notions, giving a fairness guarantee for every realization of the randomness of the mechanism). Also, every agent gets his proportional share in expectation. 

The problem in using the prioritized egalitarian mechanism when agents have $\epsilon$-leveled valuations is that it is no longer truthful
(as implied by Proposition \ref{prop:no-det-leveled}).
Yet, the mechanism is $\epsilon$-truthful,
meaning that a truthful agent loses only a small fraction 
	of his value due to being truthful. 
We suggest a (rather general) randomized method, based on combining an almost-truthful reasonable mechanism with another simple deterministic mechanism that
allocates either one or two random items to agents with random priorities, 
and  using the original mechanism only on the remaining items, to obtain a truthful in expectation mechanism.
Crucially, the second mechanism will create a \emph{strict} incentive
(that is large enough) to be truthful, to overcome the loss of a truthful agent in the prioritized egalitarian mechanism.

More specifically, we pick one or two random items to leave out and run the prioritized egalitarian mechanism with random priorities over agents. We then run the additional mechanism with \emph{reverse} priority order, on the set of items left aside, only allocating items to agents demanding them, and not allocating the second item to the agent that got the first, if possible. The fact that we only allocate to agents demanded items creates an incentive to not hide demanded items, while the possibility of getting the first item in the expense of the second, creates incentive not to add undesired items.
We show that when $\epsilon$ is small enough, the combined mechanism will be truthful in expectation.
Additionally, it will satisfy several important fairness properties (the reversal of the priority order in the second mechanism is crucial for that).
With this high-level description in mind, we move to formally describe the mechanism.


We first introduce a deterministic truthful and reasonable allocation mechanism $\mechx$ that given a list of items $X$ (consisting of $1$ or $2$ items), priority order $\sigma$ over the agents, and a reported demanded set $R_v$ for every agent $v$, returns an allocation $\alx$ as follows:
\begin{enumerate}
	\item If there is agent $v$ with $X[1] \in R_v$, then allocate $X[1]$ to the highest priority agent $v$ such that $X[1] \in R_{v}$, and move $v$ to have the lowest priority.
	\item 	If $|X|>1$ and $X[2] \in R_v$ for some $v$, then allocate $X[2]$ to highest priority agent $v$ such that $X[2] \in R_{v}$.
\end{enumerate}

\begin{restatable}{proposition}{claimmechx}
\label{cl:mechx}
	For agents that have additive dichotomous valuations, $\mechx$ is reasonable, truthful, EF1, and moreover, an agent never envies an agent that has lower priority than him.
\end{restatable}


We now present a truthful in expectation randomized allocation mechanism $\meche$ for settings where all agent valuations are $\epsilon$-leveled for $\epsilon < \frac{1}{nm^3}$. 
Let $\mechb$ be the  prioritized egalitarian (PE)  mechanism described in Section \ref{sec:submodular} for MRF valuations (we will use it for the special case of additive dichotomous valuations.)
Let $\meche$ be the following mechanism:

\begin{enumerate}
	\item Each agent $v$ is asked to report a set $R_v$, of the items he demands.
	\item Pick an item $x$ uniformly at random from $M$, and let $X[1]=x$.
	\item With probability $1-\frac{1}{m}$ pick an item $y\neq x$ uniformly at random and set $X[2]=y$.
	\item Let $\sigma$ be a random order of priorities over the agents.
	\item Let $\alb$ be the allocation of $\mechb$ and agents ordered by $\sigma$, and the reported sets $R_v \setminus X$ for each agent $v$.
\footnote{Since $\mechb$ is reasonable, the items in $X$ will not be allocated by $\mechb$.}
	\item Let $\alx$ be the allocation of $\mechx$ with the set of items $X$ and agents ordered by $reverse(\sigma)$, and the reported sets $R_v \cap X$ for each agent $v$.
	\item  Return $A$ where $A_v =\alb_v \cup\alx_v$ for every  agent $v$. 		
\end{enumerate}

We use the properties of $\mechb$ and of $\mechx$ to prove that  mechanism  $\meche$ satisfies multiple desired properties
(for the proof see Appendix \ref{appendix:theorem}).
\begin{restatable}{theorem}{truthfulAmost}
	\label{thm:truthfulAmost01}
	When all agents are $\epsilon$-leveled for $\epsilon<\frac{1}{nm^3}$
	 mechanism $\meche$ has the following properties:
	\begin{enumerate}
		\item \label{item:truthful2}
		$\meche$ is a truthful in expectation mechanism. I.e.,
		each agent maximizes his expected value (over the randomization of the mechanism) by reporting the  set of items he demands truthfully, for any reports of the others.\footnote{Moreover, if it is known to the agent that at least one of his demanded items is also demanded by some other agent, then being truthful is the \emph{unique} strategy that maximizes his expected utility. } Additionally, the mechanism is ex-post $\epsilon$-truthful.
		\item \label{item:welfare2}  $\meche$ is reasonable. Thus the welfare ex-post is at least $\epsap$-fraction of the maximum welfare, when agents are truthful.
		\item \label{item:fair2}  $\meche$ guarantees every truthful agent his $\epsap$-maximin share. Additionally, $\meche$ is $\epsap$-EF1 for truthful agents. 
		\item \label{item:fair2-in-exp} Each truthful agent receives his proportional  share in expectation.\footnote{For additive valuations, the proportional share is at least the maximin share.} 
		\item \label{item:poly2} $\meche$ can be implemented in  polynomial time in the number of items and agents.	
	\end{enumerate}
\end{restatable}

\section{Acknowledgment}
We are grateful to Herve Moulin for suggesting that a randomized version of our deterministic mechanism might be ex-ante envy-free, a suggestion that motivated us to add Theorem~\ref{thm:RPE-main} to our paper.

\bibliographystyle{abbrvnat}

\bibliography{bib}
\pagebreak

\appendix
\section{Approximate truthfulness and approximate fairness notions}\label{app:approx-fair}
Each of the notions of maximin share, EF, EF1 and EFX can be relaxed to hold only up to an multiplicative term of $\alpha$.
We next formally defined those notions.

For a given valuations functions $f=(f_1,f_2,\ldots,f_n)$ and $\alpha\in [0,1]$: 
\begin{itemize}
	\item An allocation $A$ is {\em $\alpha$-maximin fair} if $f_v(A_v)$ is at least $\alpha$ fraction of his maximin share for every agent $v\in V$.
	We use $\alpha$-{maximin}$(f_v)$ to denote this share.
	\item An allocation $A$ is \emph{$\alpha$-EF}  if for every $v_1,v_2 \in V$ $f_{v_1}(A_{v_1}) \geq \alpha\cdot f_{v_1}(A_{v_2})$.
	\item An allocation $A$ is \emph{$\alpha$-EF1}  if for every $v_1,v_2 \in V$ such that $A_{v_2}$ is not empty there exist an item $a \in A_{v_2}$ such that $f_{v_1}(A_{v_1}) \geq \alpha\cdot f_{v_1}(A_{v_2} \setminus \{a\})$.
 	\item An allocation $A$ is \emph{$\alpha$-EFX}  if for every $v_1,v_2 \in V$ and for every $a \in A_{v_2}$ it holds that $f_{v_1}(A_{v_1}) \geq \alpha\cdot f_{v_1}(A_{v_2} \setminus \{a\})$.
 	\item A mechanism $\mech$ is $\epsilon$-truthful if for every agent $v$ with valuation $f_v$ and any reports of the other agents $f_{-v}$, and any report $f_v'$ of $v$, it holds that $$(1+\epsilon) f_v(\mech(f_v,f_{-v})_v) \geq  f_v(\mech(f_v',f_{-v})_v). $$
 	I.e., agent $v$ can only increase his value by at most a multiplicative factor of $1+\epsilon$ by being non-truthful.
\end{itemize}

\section{Missing proofs for MRF Valuations}
\label{appendix:MRF}

\subsection{Fairness Properties of Lorenz Dominating Allocations}

We next restate and prove Proposition \ref{pro:LorenzIsGood}.

\LorenzIsGood*

\begin{proof}
	Let $f=(f_1,\ldots, f_n)$ denote the valuation functions of the agents, and suppose that $A = (A_1, \ldots, A_n)$ is a Lorenz dominating allocation.
	
	\begin{enumerate}
		
		\item Let $A'$ be an allocation that maximizes welfare. Then for the sorted utility vectors, $s_{A,f} \ge_{Lorenz} s_{A',f}$ implies (by taking $k=n$) that $\sum_{i=1}^n f_i(A_i) \ge \sum_{i=1}^n f_i(A'_i)$. Consequently, $A$ maximizes welfare.
		
		\item Allocation $A$ is also a lexmin allocation, because for every allocation $A'$, the inequality $s_{A,f} \ge_{lexmin} s_{A',f}$ is implied by $s_{A,f} \ge_{Lorenz} s_{A',f}$.
		
		\item Allocation $A$ can be shown to maximize NSW by considering the entries of $s_{A,f}$ in a forward order, and using concavity of the product function $\prod_i x_i$ (for non-negative variables).
		
		\item Allocation $A$ maximizes welfare (as required by min-square allocations), and can be shown to be min-square by considering the entries of $s_{A,f}$ in a backward order, and using convexity of the min-square function $\sum_i (x_i)^2$.
		
	\end{enumerate}
	
	Suppose now that the valuations are MRFs and that the Lorenz dominating allocation is non-redundant. In this case, we prove that the allocation is EFX. Let $i,j \in [n]$ be such that {$i$ envies $j$, that is} $f_i(A_j) > f_i(A_i)$. By the non-redundancy property we have that $f_i(A_i) = |A_i|$. If $|A_j| \le |A_i| + 1$, then the EFX condition holds. Hence suppose for the sake of contradiction that $|A_j| \ge |A_i| + 2$. By the non-redundancy property we have that $f_j(A_j) = |A_j| \ge f_i(A_i) + 2$. By the matroid exchange property of MRFs it follows that there is an item $e \in A_j$ such that $f_i(A_i \cup \{e\}) = f_i(A_i) + 1$ (and $f_j(A_j \setminus \{e\}) = f_j(A_j) - 1$). Moving item $e$ from agent $j$ to agent $i$ gives an allocation that Lorenz dominates $A$, thus contradicting the assumption that $|A_j| \ge |A_i| + 2$.

	Finally, suppose that the valuations are additive dichotomous. In this case, we prove that any Lorenz dominating allocation $A$ is maximin fair. Let $D_v$ be the set of demand items for agent $v$. The maximin share of $v$ is $t = \left\lfloor \frac{|D_v|}{n} \right\rfloor$. Assume for the sake of contradiction that $v$ receives at most $t-1$ items from $D_v$.
	As $A$ is welfare maximizing, all of $D_v$ is allocated, and items in $D_v$ are only allocated to agents that demand these items.
	Hence some other agent, say $u$, receives at least $t+1$ items from $D_v$, and has utility at least $t+1$. Moving one such item from $u$ to $v$ results in an allocation that Lorenz dominates $A$. As no allocation can Lorenz dominate the Lorenz dominating allocation $A$, it must be that $v$ receives at least his maximin share.
\end{proof}

We next restate and prove Proposition \ref{pro:LorenzIsBad}.

\LorenzIsBad*

\begin{proof}
	Consider a set $M$ of $m = (2n-1)n$ items that is partitioned into two sets: a set $G$ with $(n-1)n$ items and a set $B$ with $n^2$ items. Define the MRF $f_1$ as $f_1(S) = |S \cap G| + \min[|S \cap B|,n]$. For $2 \le j \le n$, define the MRF $f_j$ as $f_j(S) = |S \cap G|$. The maximin share of agent~1 is $2n-1$ (partition $M$ into $n$ bundles, each containing $n-1$ items from $G$ and $n$ items from $B$). However, every Lorenz dominating allocation gives each agent $n$ items (agent~1 gets $n$ items from $B$, and the other agents each get $n$ items from $G$). Hence Lorenz domination does not imply maximin fairness.
	
	Now let $f = \{f_1, \ldots, f_n\}$ be a collection of arbitrary MRF valuation functions and let $A = \{A_1, \ldots, A_n \}$ be an arbitrary Lorenz dominating allocation.
	Suppose that the maximin share of agent~1 is $t$. This means that there is a partition of the items into $n$ bundles $(S_1, \ldots , S_n)$, with $f_1(S_j) \ge t$ for all $1 \le j \le n$. 
	Denote $f_1(A_1)$ by $t'$, and suppose for the sake of contradiction that $t' < \frac{t}{2}$. Then the matroid exchange property implies that for every $1 \le j \le n$, the set $S_j$ contains at least $t - t' \ge t' + 1$ distinct items such that each one of them has marginal value~1 to $f_1$ with respect to $A_1$. Call these items {\em valuable} and note that the total number of valuable items is at least $n(t'+1)$. No agent $i \not= 1$ can hold more than $t' + 1$ valuable items, as the Lorenz dominating allocation would transfer such an item from agent $i$ to agent~1. Hence there are at most $(n-1)(t'+1)$ valuable items allocated to other agents, which implies that some valuable item remains unallocated. This contradicts the assumption that $A$ is Lorenz dominating.  
\end{proof}

\subsection{Missing proofs from Section \ref{sec:Lorntz-exists}}
\label{sec:PE-mech-appendix}

In the proof of Theorem~\ref{thm:LorenzDuttaRay} we made use of Lemma~\ref{cor:GStoSubmodular}. Before proving that Lemma, we introduce a key lemma that will serve us in several of our proofs. The lemma basically shows that we can move from one allocation closer (with respect to utilities) to the other, by ``moving an item" from an agent that got too much, to one that got too little. The item that one agent lost  might not be the same as the item another agent gains, and in  the process we might need to exchange some items between the agents.

\begin{lemma}
	\label{lem:GraphArgument}
	Let $f = (f_1, \ldots, f_n)$ be MRF valuations, and let $A = (A_0, A_1, \ldots, A_n)$ and $B = (B_0, B_1, \ldots, B_n)$ be two non-redundant allocations, where $A_0$ and $B_0$ specify the sets of items that remain unallocated.
	Let $S^+ \subseteq \{0, 1, \ldots, n\}$ be the set of those indices $j$ for which $|A_j| > |B_j|$, let $S^-$ be the set of indices $j$ for which $|A_j| < |B_j|$,  and let $S^=$ be the set of indices $j$ for which $|A_j| = |B_j|$. Suppose that $S^- \setminus \{0\}$ is not empty (consequently also $S^+$ is not empty), and let $i$ be an agent with $i\in S^-$. 
	Then there is a non-redundant allocation $C = (C_0, C_1, \ldots, C_n)$ with the following properties:
	
	\begin{enumerate}
		
		\item $|C_i| = |A_i| + 1$.
		
		\item There is precisely one index $k \in S^+$ for which $|C_k| = |A_k| - 1$.
		
		\item For every $j \not\in \{i,k\}$ it holds that $|C_j| = |A_j|$.
		
		\item For every $j \in S^+$ it holds that $C_j \subseteq A_j$ (with set equality unless $j = k$).
		
	\end{enumerate}
	
\end{lemma}

Before proving the lemma, let us provide some intuition. Consider three agents $p_1, p_2, p_3$ and two items $e_1,e_2$, where agents have additive dichotomous valuations, $p_1$ desires only $e_1$, and $p_2$ and $p_3$ both desire both items. {We consider allocations that leave no item unallocated.
	Let $A = (A_1,A_2,A_3)=(\phi, \{e_1\}, \{e_2\})$ and $B = (B_1,B_2,B_3)= (\{e_1\}, \{e_2\}, \phi)$.} 
Then $S^+ = \{p_3\}$, $S^= = \{p_2\}$ and $S^- = \{p_1\}$. Hence only $p_1$ can serve as $i$ in Lemma~\ref{lem:GraphArgument}, only $p_3$ can serve as $k$, and only $B$ can serve as $C$ (by simple case analysis.) Observe that transforming $A$ to $C = B$ involves not only $p_3$ giving up an item and $p_1$ receiving an item, but also $p_2$ exchanging an item.

In general, when valuation functions are additive dichotomous, it can be shown that transforming from $A$ to $C$ as required by Lemma~\ref{lem:GraphArgument} can be achieved by agent $i$ gaining an item, agent $k$ losing an item, and any other agent in $S^{-} \cup S^=$ exchanging at most one of his items. However, when valuations are submodular dichotomous (MRF), this is no longer true.

Consider four agents $p_1, p_2, p_3, p_4$ and four items $e_1,e_2,e_3,e_4$. Agents $p_1$, $p_3$ and $p_4$ have additive dichotomous valuations, $p_1$ desires only $e_1$, $p_4$ desires only $e_4$, and $p_3$ desires items $e_2$ and $e_3$. The valuation function $f_2$ of $p_2$ is MRF, with $f_2(S) = |S \cap \{1,3\}| + \min[1, |S \cap \{2,4\}|$. 
{We consider allocations that leave no item unallocated.}
Let $A = (\phi, \{e_1, e_2\}, \{e_3\}, \{e_4\})$ and $B = (\{e_1\}, \{e_3, e_4\}, \{e_2\}, \phi)$. Then $S^+ = \{p_4\}$, $S^= = \{p_2, p_3\}$ and $S^- = \{p_1\}$. Hence only $p_1$ can serve as $i$ in Lemma~\ref{lem:GraphArgument}, only $p_4$ can serve as $k$, and only $B$ can serve as $C$ ({by simple case analysis, when observing} that $f_2(\{e_2,e_4\}) = 1 < 2$). Transforming $A$ to $C = B$ involves changing the allocation of all items. In particular, $p_2$ exchanges both of his items.

We now proceed with the proof of Lemma~\ref{lem:GraphArgument}.

\begin{proof}(Lemma~\ref{lem:GraphArgument}) 
	For an item $e$, let $A(e)$ denote the agent that $e$ is allocated to under $A$ (and 0 if $e$ is not allocated), and let $B(e)$ denote the agent that $e$ is allocated to under $B$.
	
	Consider the following labeled directed (multi-) graph $G_{A \rightarrow B}$ (it may have parallel edges) with nodes $v_0, v_1, \ldots v_{n}$.  For every item $e$, if $B(e) \not = A(e)$ then place a directed edge $(v_{A(e)},v_{B(e)})$ and label it by $e$. Nodes in $S^-$ have higher in-degree than out-degree, nodes in $S^=$ have the same in-degree as out-degree, and nodes in $S^+$ have higher out-degree than in-degree. A directed (not necessarily simple) path in the graph will simply be referred to as a path. A path will be called {\em legal} if starting at allocation $A$ and transfer those items that label the edges of the path (each such item is transfered from the agent who holds it under $A$ to the agent who holds it under $B$) results in a non-redundant allocation. Recall that $i \in S^-$ and that $S^+$ is nonempty. A legal path will be called {\em useful} if it starts at $S^+$ (let $v_k$ denote its starting vertex), then never visits $S^+$ again, and ends at $v_i$. A useful path must exist, by the following inductive argument that constructs a useful path by starting at the end of the useful path (at an edge entering $v_i$) and working backwards towards the beginning of the path (to an edge leaving $v_k$). 
	
	Start at $v_{i}$. As $v_{i}\in S^-$, we have that $|B_i| > |A_i|$. The matroid exchange property implies that there must be at least $|B_i| - |A_i|$ different items such that if we transfer any of them from the agent holding it under $A$ to agent $i$, the allocation remains non-redundant. Choose one such item $e$. ({\bf Remark.} For the purpose of proving the lemma, $e$ can be chosen arbitrarily. However, when we use this Lemma in the proof of Lemma~\ref{lem:strongfaith}, we shall choose the item $e$ in a more careful way.) Item $e$ necessarily labels an incoming edge into $v_i$, say from vertex $v_j$. Transfer $e$ from $v_j$ to $v_i$ (and include the edge labeled by $e$ in the useful path). This changes $A$ into a new allocation $A'$. Now consider $v_j$. If $v_j \in S^+$ we are done. Hence it remains to address the case that $v_j \not\in S^+$. 
	For this we consider the labeled directed graph $G_{A' \rightarrow B}$, which is obtained from $G_{A \rightarrow B}$ by removing the edge labeled by $e$.
	Analogously to the definition with respect to $A$, we now have new sets $\hat{S}^{+}$, $\hat{S}^{-}$, $\hat{S}^{=}$ with respect to $A'$. Observe that $v_j \in \hat{S}^{-}$  (because $v_j \in S^= \cup S^-$ and it lost an item), that $v_i \in \hat{S}^{=} \cup \hat{S}^{-}$ (because $v_i \in S^-$ and it gained an item), and all other agents remain in their original sets (in particular, $\hat{S}^{+} = S^+$). Hence now the argument can be repeated from $v_j \in \hat{S}^{-}$. 
	Eventually, we must reach a vertex in $S^+$, as the number of edges decreases in each iteration, and there always is at least one edge incident with $S^+$.
	
	Doing all the transfers implied by the edges of the useful path gives the desired allocation $C$.
\end{proof}

Given MRF valuation functions, consider the following {\em greedy algorithm} for generating an allocation. Fix an arbitrary priority order among agents, say from~1 to $n$, and for all $1 \le i \le n$, let $S_i$ denote the set of the first $i$ agents. 
Each agent in his turn is allocated the largest possible number $n_i$ of items subject to the constraint that there is a non-redundant allocation $A^i = (A_1^i, \ldots, A_i^i)$ such that $|A_j^i| = n_j$ for every $1 \le j \le i$. 
That is, agent $i$ gets the maximum possible number of (non-redundant) items, subject to preserving the utilities of all agents that precede $i$.

\begin{lemma}
	\label{lem:greedyMRF}
	If agents have MRF valuations, then for every $i \in [n]$, the above greedy algorithm gives an allocation that attains $W(S_i)$ (maximizes welfare for the set of first $i$ agents in the priority order).
\end{lemma}

\begin{proof}
	Given the priority order over agent, let $B$ be an allocation produced by the greedy algorithm. Among all allocations that maximize welfare, let $A$ be an allocation whose utility vector (sorted according to the priority order) is lexicographically largest. We claim that $A$ and $B$ have the same utility vector. For the sake of contradiction, suppose otherwise. Then by the greedy choice of $B$, there must be an index $\ell$ such that for all $i < \ell$ we have $|A_i| = |B_i|$, and $|A_{\ell}| < |B_{\ell}|$. Apply Lemma~\ref{lem:GraphArgument} with $i = \ell$. This causes $|A_{\ell}|$ to increase by one, and the agent $k$ who loses an item must have index larger than $\ell$ (as $\{1, \ldots, k-1\} \in S^=$). This contradicts the choice of $A$ as lexicographically largest.
\end{proof}

We can now prove Lemma~\ref{cor:GStoSubmodular} that was used in the proof of Theorem~\ref{thm:LorenzDuttaRay}.

\begin{lemma}
	\label{cor:GStoSubmodular}
	For the valuation function $f' = (f'_1, \ldots, f'_n)$ defined above (MRFs augmented with the auxiliary items), the respective welfare function $W$ (defined in the proof of Theorem~\ref{thm:LorenzDuttaRay}) is submodular.
\end{lemma}

\begin{proof}
	Let $V$ denote the set of agents. To show that $W$ is submodular, we need to show that for every set $S \subset V$ and every two agents $u,v \subset V \setminus S$ it holds that $W(S \cup \{u\}) + W(S \cup \{v\}) \ge W(S) + W(S \cup \{u,v\})$. For the purpose of proving this inequality, we may ignore the auxiliary items, as their contribution to the utility is additive, and hence they contribute equally to both sides of the inequality.
	
	Fix a priority order over $S \cup \{u,v\}$ such that agents in $S$ appear first, {then $u$ followed by $v$.} 
	Apply the greedy algorithm on this order to find an allocation $A = (A_S, A_u, A_v)$ that attains $W(S \cup \{u,v\})$. By Lemma~\ref{lem:greedyMRF}, the allocation $A_S$ attains $W(S)$. As $(A_S, A_u)$ is a non-redundant allocation for $S \cup \{u\}$ and $(A_S, A_v)$ is a non-redundant allocation for $S \cup \{v\}$ we get:
	
	\begin{eqnarray*}
		W(S \cup \{u\}) + W(S \cup \{v\}) &\ge& (|A_S| + |A_u|) + (|A_S| + |A_v|) \\
		&=& |A_S| + (|A_S| + |A_u| + |A_v|) \\
		&=& W(S) + W(S \cup \{u,v\})
	\end{eqnarray*}
	proving the corollary.
\end{proof}

\subsection{Missing proofs from Section \ref{sec:PE-mech}}

\begin{proposition}
	\label{pro:faithful}
	The PE mechanism is faithful.
\end{proposition}

\begin{proof}
	Let $A = (A_1, \ldots, A_n)$ be a non-redundant Lorenz dominating allocation under $f = (f_1, \ldots, f_n)$. Replacing $f_i$ by $f_{i|A_i}$ does not enlarge the set of non-redundant allocation, and so $A$ remains Lorenz dominating. By Theorem~\ref{thm:LorenzDuttaRay}, all Lorenz dominating allocations have the same value vector, and hence agent $i$ must receive $A_i$ in the new non-redundant Lorenz dominating allocation.
\end{proof}


\begin{lemma}
	\label{lem:monotone}
	The PE mechanism is monotone.
\end{lemma}

\begin{proof}
	Fix the valuation functions $f = (f_1, \ldots, f_n)$, an agent $p$, and two sets of items $S$ and $T$ with $S \subset T$. Every agent $j \not= p$ reports $f_j$. Let $A$ be the allocation output by the PE mechanism when agent~$p$ reports $f_{p|S}$, and let $B$ be the allocation when agent~$p$ reports $f_{p|T}$. To prove monotonicity in general, we may assume that $|T| = |S|+1$ (and use induction if $|T| > |S|+1$). Hence $T$ differs from $S$ by one item, and let us call this item $a$. We need to show that $f_{p|T}(B_p) \ge f_{p|T}(A_p)$. Given that the allocations produced by PE are non-redundant, this translates to proving that $|B_p| \ge |A_p|$. Assume for the sake of contradiction that $|B_p| < |A_p|$. In this case, necessarily $a \in B_p$ (because otherwise $B$ could be output by PE instead of $A$ when agent $p$ reported $f_{p|S}$, and as $A$ and $B$ have different utility vectors, at least one of them is not Lorenz dominating).

	Let $S^+$ be the set of those indices $j$ for which $|A_j| > |B_j|$, let $S^-$ be the set of indices $j$ for which $|A_j| < |B_j|$,  and let $S^=$ be the set of indices $j$ for which $|A_j| = |B_j|$. Observe that $p \in S^+$ and hence $S^+$ is non-empty. As $B$ allocates at least as many items as $A$ (since it is welfare maximizing), then $S^-$ is non-empty as well. Among all agents in $S^+$, let $q$ denote the unique agent that minimizes $|B_q|$, breaking ties in favor of higher priority agents.  Likewise, among all agents in $S^-$, let $r$ denote the unique agent that minimizes $|A_r|$, breaking ties in favor of higher priority agents. Now there are two cases to consider, and each of them leads to a contradiction.
	
	\begin{itemize}
		
		\item Either $|A_r| < |B_q|$, or $|A_r| = |B_q|$ and $r$ has higher priority than $q$. In this case,  Lemma~\ref{lem:GraphArgument} (with $i=r$) implies that we can transform $A$ into an allocation in which $r$ gains an item, and an agent $k$ in $S^+$ loses an item. Note that even after losing an item, $k$ has at least $|B_k|$ items, implying (together with the condition defining the case) that the new allocation Lorenz dominates $A$ (with respect to the priority order of the PE mechanism). This contradicts the assumption that $A$ was chosen by the PE mechanism.
		
		\item Either $|A_r| > |B_q|$, or $|A_r| = |B_q|$ and $r$ has lower priority than $q$. In this case switch the roles of $A$ and $B$ (and of $S^+$ and $S^-$) and apply Lemma~\ref{lem:GraphArgument} with $i = q$. The lemma transforms $B$ into an allocation in which $q$ gains an item, and some agent $k$ (in the original $S^-$) loses an item, and this new allocation Lorenz-dominates $B$.
		
	\end{itemize}
	
\end{proof}

\begin{lemma}
	\label{lem:strongfaith}
	The PE mechanism is strongly faithful.
\end{lemma}

\begin{proof}
	Fix the MRF valuation functions of all agents, and let $B_v$ be the allocation to agent $v$ who reported $f_v$. By faithfulness (Proposition~\ref{pro:faithful}), if $v$'s report changes to $f_{v|B_v}$ then the allocation to $v$ remains $B_v$. Let $B = (B_1, \ldots, B_n)$ be the full allocation at this point. Now consider a set $S \subset B_v$ with $|S| = |B_v|-1$. To prove strong faithfulness it suffices to show that the allocation to $v$ when $v$ reports $f_{v|S}$ is $S$ (and then induction implies the same when $|S| < |B_v|-1$). 
	
	Let $A = (A_1, \ldots, A_n)$ be the full allocation when $v$ reports $f_{v|S}$.  Assume for the sake of contradiction that $A_v \not= S$. Then by non-redundancy, $|A_v| \le |B_v| - 2$. Let $e$ denote the unique item in $B_v \setminus S$. Consider the sets $S^+$, $S^-$ and $S^=$ as in Lemma~\ref{lem:GraphArgument}, and observe that $v \in S^-$. As the number of items allocated under $B$ is at most one more than under $A$ (and recall that $|A_v| \le |B_v| - 2$), $S^+$ includes at least one agent. Now we complete the proof via an argument copied almost verbatim from the proof of Lemma~\ref{lem:monotone}, except for a small change in the first of the two cases below.
	
	Among all agents in $S^+$, let $q$ denote the unique agent that minimizes $|B_q|$, breaking ties in favor of higher priority agents.  Likewise, among all agents in $S^-$, let $r$ denote the unique agent that minimizes $|A_r|$, breaking ties in favor of higher priority agents. Now there are two cases to consider, and each of them leads to a contradiction.
	
	\begin{itemize}
		
		\item Either $|A_r| < |B_q|$, or $|A_r| = |B_q|$ and $r$ has higher priority than $q$. In this case,  Lemma~\ref{lem:GraphArgument} (with $i=r$) implies that we can transform $A$ into an allocation in which $r$ gains an item, and an agent $k$ in $S^+$ loses an item. Importantly (see the remark in the proof of Lemma~\ref{lem:GraphArgument}), we can do so without transferring item $e$ to agent $v$ (recall that in the setting of $A$, agent $v$ does not desire item $e$), because $|B_v| \ge |A_v| + 2$.  Note that even after losing an item, $k$ has at least $|B_k|$ items, implying (together with the condition defining the case) that the new allocation Lorenz dominates $A$ (with respect to the priority order of the PE mechanism). This contradicts the assumption that $A$ was chosen by the PE mechanism.
		
		\item Either $|A_r| > |B_q|$, or $|A_r| = |B_q|$ and $r$ has lower priority than $q$. In this case switch the roles of $A$ and $B$ (and of $S^+$ and $S^-$) and apply Lemma~\ref{lem:GraphArgument} with $i = q$. The lemma transforms $B$ into an allocation in which $q$ gains an item, and some agent $k$ (in the original $S^-$) loses an item, and this new allocation Lorenz-dominates $B$.
		
	\end{itemize}
\end{proof}

\begin{remark}
	The combination of strong faithfulness and monotonicity implies a Lipschitz property for PE. That is, if agent $i$ changes its report from $f_{i|S}$ to $f_{i|(S \cup \{a\})}$ (for an MRF $f_i$, a set $S$, and an item $a \not\in S$), then the number of items allocated to $i$ increases by at most~1 (and does not decrease, by monotonicity). To see this, suppose that for $f_{i|S}$ the allocation to agent $i$ is $A_i$, for $f_{i|(S \cup \{a\})}$ the allocation is $A'_v$, and $|A'_v| \ge |A_v| + 2$. By faithfulness, reporting $f_{i|A'_v}$ the allocation remains $A'_v$. By strong faithfulness, reporting $f_{i|(A'_v \setminus \{a\})}$ gives the allocation $A'_v \setminus \{a\}$, and this allocation contains more items than $A_v$. As $(A'_v \setminus \{a\}) \subset S$, this contradicts monotonicity.
\end{remark}

\subsection{Polynomial Time Algorithm for the PE Mechanism}
We next restate and prove Theorem \ref{thm:MRF-poly-time}.

\MRFpolytime*

\begin{proof}
	Given a priority order $\sigma$ (w.l.o.g., from~1 to $n$) over the agents and MRF valuation functions $f_1, \ldots, f_n$, we reduce the problem of finding a Lorenz dominating allocation (with respect to $\sigma$) to a polynomial sequence of matroid intersection problems.

	The matroid intersection problem has a set $S$ of $mn$ items, arranged in $n$ groups, $S_1, \ldots S_n$, each with $m$ items. For every item $e\in M$, make $n$ copies $e_1, \ldots, e_n$ of the item, placing $e_i$ in $S_i$ for every $i$. We now define two matroids over $S$. A set $T$ is independent in matroid $M_1$ if and only if for every $i\in [n]$ the set $T \cap S_i$ is independent with respect to the MRF $f_i$. A set is independent in matroid $M_2$ if for each of the original items $e\in M$, the set $T$ contains at most one of its copies. There is a natural bijection between the set of non-redundant allocations and the independent sets of $M_1 \cap M_2$. In this bijection, an allocation $A$ is mapped to the set $T$ that is independent in $M_1 \cap M_2$, where $e_j \in T$ iff item $e$ is allocated to agent $j$ under $A$. Hence allocations that maximize welfare correspond to independent sets of maximum rank. These can be found by a matroid intersection algorithm~\cite{edmonds2003submodular}. Matroid intersection algorithms can be run in polynomial time if one is given oracle access to {\em independence queries} for each of the two matroids $M_1$ and $M_2$. Having a succinct representation for the MRFs (that allows answering value queries) suffices for this purpose.
	
	Let $A = (A_1, \ldots, A_n)$ be an allocation that maximizes welfare, found by the matroid intersection algorithm. As $A$ need not be Lorenz dominating, we are not done yet. To proceed,
	we use the correspondence between Lorenz dominating allocations and min-square allocations (see Proposition~\ref{pro:LorenzIsGood}). Specifically, given the priority order, it can be seen that the Lorenz dominating allocation is the one that minimizes the {\em potential} $\sum_{i\in [n]} (|A_i| + \frac{i}{n})^2$.
	
	We now present a polynomial time algorithm that tests whether $A$ is an allocation of smallest potential (among allocations that maximize welfare), and if the outcome is negative, it returns a maximum welfare allocation $C$ with smaller potential. We shall use Lemma~\ref{lem:GraphArgument}. Suppose that $A$ does not minimize potential, and let $B$ denote an arbitrary Lorenz dominating allocation (and hence it has smallest potential). Let $i \in S^-$ be the agent with smallest $|A_i|$ (breaking ties in favor of agents of higher priority). Let $j \in S^+$ be the agent with smallest $|B_i|$ (breaking ties in favor of agents of higher priority). Necessarily $|A_i| \le |B_j|$, with equality only if $i$ has higher priority than $j$, as otherwise $B$ cannot be Lorenz dominating. Applying Lemma~\ref{lem:GraphArgument}, there is an allocation $C$ whose vector of utilities differs from that of $A$ in two entries: $i$ gains one item, and an agent $k \in S^+$ loses one item. Necessarily, $C$ has smaller potential than $A$. To find such an allocation $C$ (note that neither $B$ nor the sets $S^-$ and $S^+$ are known to the algorithm), try all $O(n^2)$ choices for $(i,k)$ that make sense (e.g., no need to consider pairs for which $|A_i| > |A_k|$). The feasibility of each such choice can be checked as follows.
	Observe that for every nonnegative integer $t$ and every MRF $f_i$ the function $f_{i|\le t}$ defined as $f_{i|\le t}(T) = \min[t, f_i(T)]$ (for every set $T \subset M$), is still an MRF. Solve the matroid intersection problem with the following modified MRF valuation functions $f_{1 | \le t_1}, \ldots, f_{n | \le t_n}$, where $t_i = |A_i|+1$, $t_k = |A_k| - 1$, and $t_j = |A_j|$ for every $j \in [n] \setminus \{i,k\}$. 
	
	After at most $O(m^2n^2)$ iterations an allocation of smallest potential is found. This is because the potential is a multiple of $\frac{1}{n^2}$, and the highest possible potential is at most $(m+1)^2$. 
\end{proof}

\section{Missing proofs for $\epsilon$-Leveled Valuations} 
\label{appendix:leveled}
\subsection{Basic Observations}
We next restate and prove Observation \ref{obs:dichotmous}:
\dichotomous*
\begin{proof}
	Given a set $S=\{a_1,\ldots,a_k\}$
	\begin{eqnarray*}
	\hat{f}(S) &= &  \lfloor f(S) \rfloor  = \lfloor \sum_{i=1}^k f(a_i ~ | ~ \{a_1,\ldots a_{i-1}\}) \rfloor  \\ & = &  \lfloor \sum_{i=1}^k \lfloor f(a_i ~ | ~ \{a_1,\ldots a_{i-1}\})\rfloor + \sum_{i=1}^k  \left(f(a_i ~ | ~ \{a_1,\ldots a_{i-1}\})-\lfloor f(a_i ~ | ~ \{a_1,\ldots a_{i-1}\})\rfloor \right)  \rfloor\\ 
	&=& \sum_{i=1}^k \lfloor f(a_i ~ | ~ \{a_1,\ldots a_{i-1}\})\rfloor,
	\end{eqnarray*} 
where the last equality is since the second sum is strictly less than $1$ since we sum $k\leq m$ terms, and each is at most $\epsilon< 1/m$.
Therefore: $$\hat{f}(a~|~S) = \lfloor f(a ~ | ~ S)\rfloor \in \{0,1\}$$
and thus $\hat{f}$ is dichotomous. 

If $f$ is also additive, then so is $\hat{f}(a)$.
\end{proof}
\cout{
We observe by the proof of Observation \ref{obs:dichotmous} that if $f$ is also additive (i.e. $\epsilon$-leveled) then $\hat{f}$ remains additive.
\begin{observation}\label{obs:rounding-additive}
	For $\epsilon< \frac{1}{m}$, and any $\epsilon$-leveled valuation $f$, the function $\hat{f}$ is dichotomous and additive. 
\end{observation}

We also observe by the proof of Observation \ref{obs:dichotmous} that:
\begin{observation}\label{obs:rounding-additive-loss}
	For $\epsilon< \frac{1}{m}$, and any $\epsilon$-dichotomous valuation $f$, and every set $S \subseteq M$,
	\begin{equation}
	f(S) \leq \hat{f}(S) + m\epsilon. \label{eq:dichotomous_upper}
	\end{equation} 
\end{observation}
}

\begin{observation} \label{obs:guarantees}
	Given a mechanism $\mech$ for dichotomous valuations, let $\hat{\mech}$ be the allocation mechanism for  $\epsilon$-dichotomous valuations such that for every valuations $f$ the allocation $\hat{\mech}(f)$ is defined to be the allocation $\mech(\hat{f})$. It holds that:
	\begin{enumerate}
		\item if	$\mech$ is maximin fair, then $\hat{\mech}$ is $\epsap$-maximin fair.
		\item If $\mech$ is EF1 then $\hat{\mech}$ is $\epsap$-EF1. 
		\item If $\mech$ is EFX then $\hat{\mech}$ is $\epsap$-EFX. 
		\item If $\mech$ is truthful then $\hat{\mech}$ is $\epsilon$-truthful. 
	\end{enumerate}
\end{observation}
\begin{proof}
	 Let $A$ be the allocation $\hat{\mech}(f)$.
	\begin{enumerate}
		\item We first observe that for $\epsilon <\frac{1}{m}$ and every $\epsilon$-dichotomous valuation $f_v$, it holds that:
		\begin{equation}
		\mbox{maximin}(f_v) \leq (1+\epsilon )\mbox{maximin}(\hat{f}_v). \label{eq:maximin_dichotomous}
		\end{equation}
		This follows by considering the values of $f_v$ and $\hat{f}_v$ in the maximin partition according to $f_v$ and Equation \eqref{eq:dichotomous_upper}.
		
		We have that for each agent $v$ and any possible allocation $A_v$ (when reporting $f_v$):
		\begin{equation*}
		f_v(A_v) \stackrel{\eqref{eq:dichotomous_lower}}{\geq } \hat{f}_v(A_v) \geq \mbox{maximin}(\hat{f}) \stackrel{\eqref{eq:maximin_dichotomous}}{\geq} \epsap\mbox{maximin}(f).
		\end{equation*}

		
		\item By EF1 of $\mech$ with respect to $\hat{f}$ for every agents $v_1,v_2$, if $A_{v_2}$ is not empty there exists an item $a \in A_{v_2} $ such that: 
		$$ f_{v_1}(A_{v_1}) \stackrel{\eqref{eq:dichotomous_lower}}{\geq } \hat{f}_{v_1}(A_{v_1}) \geq \hat{f}_{v_1}(A_{v_2}\setminus \{a\}) \stackrel{ \eqref{eq:dichotomous_upper}}{\geq}\epsap f_{v_1}(A_{v_2}\setminus \{a\}).$$
		\item By EFX of $\mech$ with respect to $\hat{f}$ for every agents $v_1,v_2$ and for every item $a \in A_{v_2} $ it holds that:
		$$ f_{v_1}(A_{v_1}) \stackrel{\eqref{eq:dichotomous_lower}}{\geq } \hat{f}_{v_1}(A_{v_1}) \geq \hat{f}_{v_1}(A_{v_2}\setminus \{a\}) \stackrel{ \eqref{eq:dichotomous_upper}}{\geq}\epsap f_{v_1}(A_{v_2}\setminus \{a\}) .$$
		\item By truthfulness of $\mech$ with respect to $\hat{f}$ for every allocations $A,A'$ where $A$ is the allocation under truthful reporting, and $A'$ is the allocation of the non-truthful reporting. It holds that:
		$$ f_{v_1}(A_{v_1}) \stackrel{\eqref{eq:dichotomous_lower}}{\geq } \hat{f}_{v_1}(A_{v_1}) \geq  \hat{f}_{v_1}(A'_{v_1})\stackrel{ \eqref{eq:dichotomous_upper}}{\geq} \epsap f_{v_1}(A'_{v_1}).$$ 
			
	\end{enumerate}
\end{proof}

\begin{observation}
	For $\epsilon$-dichotomous valuations $f$, every reasonable allocation $A$ gives $\epsap$-approximation to the social welfare.
	\label{obs:welfare_approx}
\end{observation}
\begin{proof}
	Let $A^{opt}$ be the welfare maximizing allocation.
	We have that: 
	\begin{equation*}
	\sum_{v} f_v(A_v) \geq \sum_{v}\hat{f}_v(A_v)   \geq \sum_{v}\hat{f}_v(A^{opt}_v) \stackrel{\eqref{eq:dichotomous_upper}}{\geq} \epsap \sum_{v}f_v(A^{opt}_v),
	\end{equation*}
	where the second inequality is by definition of reasonable $A$ maximizes the social welfare of the corresponding dichotomous valuations. 
\end{proof}

\subsection{Impossibilities for $\epsilon$-Leveled Agents }

We next restate and prove Proposition~\ref{prop:no-det-leveled}.
\nodetleveled*
\begin{proof}
	Let $M=\{a,b,c\}$, and $V=\{u,v\}$.
	We denote an $\epsilon$-leveled valuation function $f$ over $M$ as the vector  $r_f=(f(a),f(b),f(c))$. Given a vector $r$ of size 3, we denote by $f_r$ the additive function such that for every $S$, $f_r(\{a\})=r_1, f_r(\{b\})=r_2, f_r(\{c\})=r_3$. 
	For an allocation $A$ and an agent $w \in V$, $A_w$ is the set of items allocated to agent $w$
	
	Let $\mech$ be a deterministic truthful mechanism that approximates the social welfare for $\epsilon$-leveled agents.
	We look at the three allocations $A^1=\mech((1+\epsilon,0,0),(1+\epsilon,0,0)),A^2=\mech((0,1+\epsilon,0),(0,1+\epsilon,0)),A^3=\mech((0,0,1+\epsilon),(0,0,1+\epsilon))$.
	By that $\mech$ is almost welfare maximizer, it must be that $a$ (resp. $b,c$) is allocated in $A^1$ (resp. $A^2,A^3$). Thus, there exists  an agent that in at least two of these three allocations, was allocated the item he desires.
	We assume w.l.o.g. that agent $u$ is that agent and $a\in A^1_u$ and $b\in A^2_u$.
	
	Let $A^4=\mech((1+\epsilon,1,0),(1+\epsilon,0,0))$.
	By that $\mech$ is almost welfare maximizer, we get that $A^4_u$ must contain $b$, 
	and by truthfulness of agent $u$ we get that $f_{(1+\epsilon,1,0)}(A^4_u) \geq f_{(1+\epsilon,1,0)}(A^1_u) = 1+\epsilon > 1 =f_{(1+\epsilon,1,0)}(\{b\})$, where the first equality is since $a\in A^1_u$. Thus, $A^4_u$ must also contain item $a$, so $A^4_u$ contains both items, $a$ and $b$. 
	
	By symmetry, for $A^5=\mech((1,1+\epsilon,0),(0,1+\epsilon,0))$, we get that $a,b \in A^5_u$.
	Let $A^6=\mech((1+\epsilon,1,0),(0,1+\epsilon,0))$.
	By truthfulness for agent $u$, we get that 
	$f_{(1+\epsilon,1,0)}(A^6_u) \geq f_{(1+\epsilon,1,0)}(A^5_u) =  2+\epsilon > 1+\epsilon  =\max(f_{(1+\epsilon,1,0)}(\{a\}),f_{(1+\epsilon,1,0)}(\{b\}))$, therefore, both $a$ and $b$ must be in $A^6_u$.  
	
	Let $A^7=\mech((1+\epsilon,1,0),(1,1+\epsilon,0))$.
	By truthfulness for agent $v$ it must be that $0= f_{(1+\epsilon,0,0)}(A^4_v) \geq  f_{(1+\epsilon,0,0)}(A^7_v)$, therefore $a \notin A^7_u$.
	By truthfulness for agent $v$ it must be that $0= f_{(0,1+\epsilon,0)}(A^6_v) \geq  f_{(0,1+\epsilon,0)}(A^7_v)$, therefore $b \notin A^7_u$.
	Thus, allocation $A^7$ does not satisfy the fairness property for agent $v$.	 
\end{proof}

We next restate and prove Proposition \ref{prop:no-rand-Lorentz-leveled}.
\norandLorentzleveled*

\begin{proof}
	Let the set of items be $M =\{L,H\}$, and consider two agents with the same $\epsilon$-leveled valuation function $f$, satisfying $f(L)=1$ and $f(H)=1+\epsilon$.
	When both agents report their true valuation, any Lorentz dominating allocation (and likewise for lex-min, NSW, min square, and likewise for a reasonable allocation that is EFX) gives each agent one item. 
	Thus, in that case at least one agent has positive probability of getting item $L$. W.l.o.g., let agent~$1$ be that agent. If agent~$1$ reports a valuation $f_1(L)=0,f_1(H)=1+\epsilon$ and agent~2 reports truthfully, the only Lorentz dominating allocation (and likewise for the other fairness properties considered here) is that agent~1 receives item $H$ (and agent~2 receives $L$). Therefore agent~1 gains higher utility by not reporting her true valuation function.
\end{proof}

\subsection{Proof of the Main Theorem for $\epsilon$-leveled Valuations}
\label{appendix:theorem}
In this section we restate and prove Theorem~\ref{thm:truthfulAmost01}. 
We first prove a useful claim about $\mechx$:
\claimmechx*
\begin{proof}
	$\mechx$ is reasonable since every item in $X$ will be allocated if and only if it is in $R_v$ for some agent $v$, and only to an agent that demands it. 
	
	$\mechx$ is truthful for additive dichotomous agents since if an agent receives the whole set $X$ by being truthful, he cannot gain by misreporting. Else, by manipulating, the agent can either add an unwanted item, or trade a wanted item by another which does not  increase his value since his valuation is dichotomous.
	
	{$\mechx$ is EF1 since the only case where an agent receives two items, is if he is the only one that demanded item $X[2]$. For two agents $u$ and $v$  such that $u$ has higher priority than $v$, $u$ does not envy $v$ at all since the only way for $v$ to get an item that $u$ desires, is if $u$ got another item.}
\end{proof}

We next restate and prove Theorem~\ref{thm:truthfulAmost01}, showing that  mechanism  $\meche$ satisfies multiple desired properties.

\truthfulAmost*
 
\begin{proof}
	In this proof, whenever we write expectation or probability, the probability space is over the randomness of $\meche$ (i.e., the list $X$, and the priorities $\sigma$).
	We prove Item \ref{item:truthful2} (truthfulness) in Lemma \ref{lem:level-truthful}, Item \ref{item:welfare2} (reasonable)
	in Lemma \ref{lem:level-reasonable}, Item \ref{item:fair2} (ex-post fairness) in Lemma \ref{lem:level-fair}, Item \ref{item:fair2-in-exp} (proportional in expectation) in Lemma \ref{lem:level-propportional}, and Item \ref{item:poly2} (polynomial time) in Lemma \ref{lem:level-polynomial}.
 \end{proof}
\begin{lemma}\label{lem:level-truthful}
	$\meche$ is a truthful in expectation mechanism. I.e.,
	each agent maximizes his expected value (over the randomization of the mechanism) by reporting the  set of items he demands truthfully, for any reports of the others.
	Additionally, the mechanism is ex-post $\epsilon$-truthful.
\end{lemma}
\begin{proof}
We prove that agent $v$ with demanded set $D_v$, and a corresponding valuation $f_v$, 
cannot benefit from deviating and reporting $R_v$ that is not $D_v$.
Let $R_i$ be the report of  agent $i$ for $i\neq v$.
Given the reports of all agents but $v$, let $A$ (resp. $\tilde{A}$) be the allocation when agent $v$ {is being truthful (resp. reporting $R_v$).} 
We denote by $\alb$ (resp. $\alba$) the corresponding allocation  
returned by $\mechb$, and we denote by $\alx$ (resp. $\alxa$) the corresponding allocation returned by $\mechx$.

Since  $\mechb,\mechx$ are truthful with respect to additive dichotomous agents, we get that $\hat{f}_v(\alb_v) \geq \hat{f}_v(\alba_v)$, and $\hat{f}_v(\alx_v) \geq \hat{f}_v(\alxa_v)$. 
Therefore we get that:
\begin{equation}
\hat{f}_v(A_v) = \hat{f}_v(\alb_v)+\hat{f}_v(\alx_v)  \geq \hat{f}_v(\alba_v)+\hat{f}_v(\alxa_v) =\hat{f}_v(\tilde{A}_v). \label{eq:mechanism}
\end{equation}

If $D_v \cap R_i =\emptyset$ for every $i\neq v$, reporting truthfully is 
dominant since in this case agent $v$ will get $D_v$ which maximizes his value.

Else; there exist an agent $i$ and an item $j$ such that $j\in R_i \cap D_{v}$. 
By Observation \ref{obs:guarantees} we know that $\mechb$ is $\epsilon$-truthful, implying that an agent loss from being truthful is small. We show that by reporting truthfully, the allocation $\alx_v$ is significantly preferred over $\alxa_v$ in expectation (enough to overcome the loss in $\mechb$.)
Combining these two observations yields the truthfulness of $\meche$.
To prove the claim about $\mechx$ we first show that:  
\begin{lemma}\label{lem:strong_truthful_mx}
	For every agent $v$, if there is an item $j$ and an agent $i$ such that $j\in R_i \cap D_{v}$, then:
	\begin{equation}
	E[f_v(\alx_v)] \geq \frac{1}{nm^2} +\frac{E[f_v(\alxa_v)]}{1+\epsilon}. \label{eq:alx} 
	\end{equation}
\end{lemma}
\begin{proof}
	We prove this inequality by considering two cases of $R_v$.
	
	\textbf{Case 1 [Hiding desired items]:} There exists $\ell \in D_v \setminus R_v$ 
	(i.e., $v$ misreported and did not report item $\ell$ that he actually demands). Let $e_1$ be the event that $X=(\ell)$, and $v$ is the highest priority agent according to $\sigma$ (in $\mechx$), and let $\bar{e_1}$ be the complement event. 
	Then: 
	\begin{eqnarray*}
		E[f_v(\alx_v)] & \stackrel{\eqref{eq:dichotomous_lower}}{\geq} & E[\hat{f}_v(\alx_v)]  \\
		& = & E[\hat{f}_v(\alx_v) ~|~ e_1] \cdot Pr[e_1] + E[\hat{f}_v(\alx_v) ~|~ \bar{e_1}] \cdot Pr[\bar{e_1}]  \\
		& \geq & E[\hat{f}_v(\alxa_v) + 1 ~|~ e_1] \cdot Pr[e_1] + E[\hat{f}_v(\alxa_v) ~|~ \bar{e_1}] \cdot Pr[\bar{e_1}]  \\
		& = & Pr[e_1] + E[\hat{f}_v(\alxa_v)] \\
		& \stackrel{\eqref{eq:dichotomous_upper}}{\geq}  & \frac{1}{nm^2} + \frac{E[f_v(\alxa_v)]}{1+\epsilon},
	\end{eqnarray*}
	where the second inequality holds since in event $e_1$, $\hat{f}_v(\alx_v)=1$ and $\hat{f}_v(\alxa_v)=0$, and that $\mechx$ is truthful.
	Additionally, the third inequality uses the fact that $\Pr[e_1] = \frac{1}{nm^2}$ ($v$ has highest priority with probability $1/n$, independently, $|X|=1$ with probability $1/m$, conditional on $|X|=1$, $l$ is the item in $X$ with probability $1/m$).

	\textbf{Case 2 [Reporting undesired items as desired]:}
	There exists $\ell \in R_v \setminus D_v$ (i.e., $v$ added undemanded item $\ell$ to his report). Let $e_2$ be the event that $X=(\ell,j)$ (recall that $j\in R_i \cap D_{v}$), and $v$ 
	has the highest priority in mechanism $\mechx$.
	Then:
	\begin{eqnarray*}
		E[f_v(\alx_v)] & \stackrel{\eqref{eq:dichotomous_lower}}{\geq} & E[\hat{f}_v(\alx_v)]  \\
		& = & E[\hat{f}_v(\alx_v) ~|~ e_2] \cdot Pr[e_2] + E[\hat{f}_v(\alx_v) ~|~ \bar{e_2}] \cdot Pr[\bar{e_2}]  \\
		& \geq & E[\hat{f}_v(\alxa_v) + 1 ~|~ e_2] \cdot Pr[e_2] + E[\hat{f}_v(\alxa_v) ~|~ \bar{e_2}] \cdot Pr[\bar{e_2}]  \\
		& = & Pr[e_2] + E[\hat{f}_v(\alxa_v)] \\
		& \stackrel{\eqref{eq:dichotomous_upper}}{\geq}  &  \frac{1}{nm^2} + \frac{E[f_v(\alxa_v)]}{1+\epsilon},
	\end{eqnarray*}
	where the second inequality holds since in event $e_2$, $\hat{f}_v(\alx_v)=1$ and $\hat{f}_v(\alxa_v)=0$, and that $\mechx$ is truthful.
	Additionally, the third inequality uses the fact that $\Pr[e_2] = \frac{1}{nm^2}$ ($v$ has highest priority with probability $1/n$, independently, $|X|=2$ with probability $(m-1)/m$, and $X=(l,j)$, conditional on $|X|=2$, with probability $1/(m(m-1))$).
\end{proof}
With Lemma \ref{lem:strong_truthful_mx} we now prove that $\meche$ is truthful in expectation.
It holds that:
\begin{eqnarray*}
	E[f_v(A_v)] & = & E[f_v(\alb_v)] + E[f_v(\alx_v)] \\ 
	& \stackrel{\eqref{eq:alx}}{\geq} & \frac{E[f_v(\alba_v)]}{1+\epsilon}  +  \frac{E[f_v(\alxa_v)]}{1+\epsilon} +\frac{1}{nm^2}\\
	& > & \frac{E[f_v(\tilde{A}_v)]}{1+\epsilon}+\frac{\epsilon \cdot E[f_v(\tilde{A}_v)]}{1+\epsilon}\\
	& = & E[f_v(\tilde{A}_v)],
\end{eqnarray*} 
where the first inequality is since that $\mechb$ is $\epsilon$-truthful by Observation \ref{obs:guarantees}, and
the second inequality is since $E[f_v(\tilde{A}_v)] \leq (1+\epsilon)m$ and $\epsilon <\frac{1}{nm^3}$.

We now prove that $\meche$ is ex-post $\epsilon$-truthful.
Since both $\mechb$ and $\mechx$ are (ex-post) truthful for additive dichotomous agents, it follows immediately that $\meche$ is (ex-post) truthful for additive dichotomous agents.
Combining with Observation \ref{obs:guarantees} we get that $\meche$ is ex-post $\epsilon$-truthful. 
\end{proof}

\begin{lemma}\label{lem:level-reasonable}
	$\meche$ is reasonable. Thus the welfare ex-post is at least $\epsap$-fraction of the maximum welfare, when agents are truthful.
	\end{lemma}
\begin{proof}
$\meche$ is reasonable since both $\mechb,\mechx$ are reasonable and since the agents are additive. 
If $A^1$ (resp. $A^2$) is a reasonable allocation for the set of items $M_1$ (resp. $M_2$), then $A^1 \cup A^2$ is reasonable for the set of items $M_1 \sqcup M_2$.

The guarantee of the social welfare follows by Observation \ref{obs:welfare_approx}.
\end{proof}

\begin{lemma}\label{lem:level-fair}
	$\meche$ guarantees every truthful agent his $\epsap$-maximin share. Additionally, $\meche$ is $\epsap$-EF1 for truthful agents.
\end{lemma}
\begin{proof}
 Assume that agent $v$ is truthful, that is $R_v=D_v$.
	By Observation \ref{obs:guarantees} it is enough to show that $\meche$ is maximin fair and EF1 for additive dichotomous valuations.
	
	Let $i$ be the priority of agent $v$ {in $\mechb$}.
	Under the reports of $R$, we denote by $A$ (resp. $\alb,\alx$) the allocation $\meche$ (resp. $\mechb,\mechx$).
	Note that if $i=n$ and $|R_v \cap X |\geq 1$ or $i=n-1$ and $|R_v \cap X| \geq 2 $ then $\alx_v \neq \emptyset$, therefore:
	\begin{equation}
	\hat{f}_v(\alx_v) \geq \left\lfloor\frac{i+|X\cap R_v|-1 }{n} \right\rfloor \label{eq:x_value}
	\end{equation}
	
	We next present some properties of $\mechb$ that we use in the proof that $\meche$ is EF1 and maximin fair.
	\begin{observation}
		In $\mechb$, an agent with MRF valuation does not envy agents with lower priority. \label{obs:envy}
	\end{observation}
	\begin{proof}
		Assume towards contradiction that an agent $v$ envies a lower priority agent $u$.
		Since $\mechb$ is not redundant, it means that $|\alb_v| = \hat{f}_v(\alb_v)  < \hat{f}_v(\alb_u) = |\alb_u \cap R_v|$.
		By the exchange property of matroids, there exists an item in $\alb_u \cap R_v$, such that transferring this item from $u$ to $v$, will be a non-redundant allocation in which $v$ gets higher value, and therefore will Lorenz-dominates $\alb$ which contradicts that $\mechb$ returns a Lorenz-dominating allocation with respect to the priority order.
	\end{proof}

	We now give a tighter bound of the number of items allocated by $\mechb$ to an agent depending on his priority.
	\begin{claim}\label{claim:PE-priority}
		For agents that report additive dichotomous valuations, for every $i\in [n]$, mechanism $\mechb$ allocates at least $ \left\lceil\frac{|R_v\setminus X|-i+1}{n}\right\rceil$ items  
		to the agent $v$ that  has priority $i$. 
	\end{claim}
	
	\begin{proof} 
		Fixing the set $X$,  let $n_v$ be the number of items allocated by $\mechb$ to agent $v$.
		Since $\mechb$ is EF1, for every agent $u$ with higher priority than $v$ (there are $i-1$ such agents), the number of items allocated to $u$ among $R_v\setminus X$ is at most $n_v+1$.
		{By Observation \ref{obs:envy},} in $\mechb$ agent $v$ does not envy agent with lower priority than $v$ (there are $n-i$ such agents), the number of items allocated to $u$ among $R_v\setminus X$ is at most $n_v$.
		Since $\mechb$ is welfare maximizer, it must be that all items in $R_v$ are allocated, thus
		$ (i-1)\cdot (n_v+1)  +n_v +(n-i)\cdot n_v \geq |R_v \setminus X|$, which implies the claim.
	\end{proof}
	
	With Observation \ref{obs:envy} and Claim \ref{claim:PE-priority} we can prove that $\meche$ is maximin fair and EF1 for additive dichotomous agents.
	
	We first prove that $\meche$ is maximin fair for additive dichotomous valuations.
	Consider a agent $v$.
	It holds that 
	\begin{equation}
	\mbox{maximin}(\hat{f}_v)
	=\left\lfloor\frac{|R_v|}{n}\right\rfloor.
	\label{eq:max_min_add}
	\end{equation}
	
	It holds that:
	\begin{eqnarray}
	\hat{f}_v(A_v)  &= & \hat{f}_v(\alb_v)+\hat{f}_v(\alx_v) \nonumber \\
	&\stackrel{\eqref{eq:x_value}}{\geq } & \left\lceil \frac{|R_v\setminus X|-i+1}{n}\right\rceil + \left\lfloor\frac{i+|X\cap R_v|-1 }{n} \right\rfloor \nonumber \\
	&\geq & \left\lfloor\frac{|R_v|}{n}\right \rfloor \stackrel{\eqref{eq:max_min_add}}{=} \mbox{maximin}(\hat{f}),
	\end{eqnarray}
	{where the first inequality is by Inequality~(\ref{eq:x_value}) and 
		Claim \ref{claim:PE-priority}, }and the last inequality is by the fact that $\lceil \frac{a}{n}\rceil +\lfloor \frac{b}{n}\rfloor \geq \lfloor \frac{a+b}{n} \rfloor$ for every integers $a,b$.

	{
		Now we prove that $\meche$ is EF1 for additive dichotomous agents.
		{We have that $\mechb$ is EF1 for additive dichotomous agents, furthermore, {by Observation \ref{obs:envy}} in the allocation $\alb$, agents with higher priority
			do not envy agents with lower priority (priorities in  $\mechb$).} 
		Let, $u,v$ be any two agents. If $u$ has higher priority than agent $v$ (in $\mechb$), then $ \hat{f}_u(\alb_u) \geq \hat{f}_u(\alb_v)$, and if $u$ has lower priority than agent $v$ (in $\mechb$) then $ \hat{f}_u(\alb_u) \geq \hat{f}_u(\alb_v)-1$.
		By Claim \ref{cl:mechx}, we have that $\mechx$ is EF1, and agents with higher priority (in $\mechx$) do not envy agents with lower priority.
		I.e.,  if $u$ has higher priority than agent $v$ (in $\mechx$), then $ \hat{f}_u(\alx_u) \geq \hat{f}_u(\alx_v)$, and if $u$ has lower priority than $v$ (in $\mechx$) then $ \hat{f}_u(\alx_u) \geq \hat{f}_u(\alx_v)-1$.
		We are now ready to prove that agent $u$  does not envy agent $v$ up to one item (i.e., $\hat{f}_u(A_u) \geq \hat{f}_u(A_v)-1$). 
		
		Since the priorities in $\mechb$ and $\mechx$ are reversed, we get that

		\textbf{Case 1:} If $u$ has higher priority than $v$ in $\mechb$ (and lower priority in $\mechx$) then   
		\begin{equation*}
		\hat{f}_u(A_u) = \hat{f}_u(\alb_u) + \hat{f}_u(\alx_u) \geq \hat{f}_u(\alb_v) +\hat{f}_u(\alx_v)  -1 = \hat{f}_u(A_v)-1.
		\end{equation*} 
		
		\textbf{Case 2:} If $u$ has lower priority than $v$ in $\mechb$ (and higher priority in $\mechx$) then   
		\begin{equation*}
		\hat{f}_u(A_u) = \hat{f}_u(\alb_u) + \hat{f}_u(\alx_u) \geq \hat{f}_u(\alb_v) -1 +\hat{f}_u(\alx_v)  = \hat{f}_u(A_v)-1.
		\end{equation*} 
		This completes the proof.
	}
	
	\end{proof}

\begin{lemma} \label{lem:level-propportional}
	 Each truthful agent receives his proportional  share in expectation.\footnote{For additive valuations, the proportional share is at least the maximin share.}
\end{lemma}
\begin{proof}
 Suppose that agent $v$ is truthful, that is $R_v=D_v$.
	We show that $v$ receives his proportional share in expectation.
	For agent $v$ and profile of reports $R=(R_1,\ldots,R_n)$, 
	if for all $i$ it holds that $R_i=R_v$ then the allocation will always be the same partition, allocated according to the permutation over priorities. Thus, agent $v$ will get each set of the partition with probability $\frac{1}{n}$ which yields an expected value of $\frac{f_v(R_v)}{n}$.
	
	Else,  there exists $i \neq v$ such that $R_i \neq R_{v}$.
	It holds that:
	\begin{equation}
	\frac{f_v(R_v)}{n} \stackrel{\eqref{eq:dichotomous_upper}}{\leq} \frac{(1+\epsilon)\hat{f}_v(R_v)}{n} = (1+\epsilon)\frac{|R_v|}{n}. \label{eq:proportional_upper}
	\end{equation}
	
	Under the reports of $R$, we denote by $A$ (resp. $\alb,\alx$) the allocation of $\meche$ (resp. $\mechb,\mechx$).
	
	{For every given list $S$ of size in $\{1,2\}$, by Claim \ref{claim:PE-priority}, given that  $X=S$, if agent $v$ has priority $i$ then $|\alb_v| \geq \left\lceil\frac{|R_v \setminus S|-i+1}{n}\right\rceil$.} 
	Thus, when considering the expectation over the priorities of $v$ for a fixed $S$: 
	\begin{equation}
	E_{\sigma}[\hat{f}_v(\alb_v) \mid X=S] \geq \frac{1}{n}\sum_{i=1}^n \left\lceil\frac{|R_v \setminus S|-i+1}{n}\right\rceil = \frac{|R_v \setminus S|}{n} = \frac{|R_v|}{n} -\frac{|R_v \cap S|}{n}, \label{eq:prop1}
	\end{equation}
	where the first equality is since for every $x=an+b$ (for $0 \leq b<n$), the number of terms in $\{\lceil \frac{x-i}{n}\rceil\}_{0 \leq i <n}$ that equals $a+1$ (resp. $a$) is $b$ (resp. $n-b$).
	
	When taking expectation also over $X$, we get that: 
	\begin{eqnarray}
	E[\hat{f}_v(\alb_v)] &\stackrel{\eqref{eq:prop1}}{\geq }&  \frac{|R_v|}{n}  -\frac{E[|R_v \cap X|]}{n}\nonumber\\
	& =  & \frac{|R_v|}{n} -\frac{|R_v| \cdot E[|X|]}{nm} = \frac{|R_v|}{n}\cdot \frac{m^2-2m+1}{m^2}. \label{eq:mb}
	\end{eqnarray}
	where for the last equality we used $E[|X|]= \frac{1}{m} + 2(1-\frac{1}{m}) = 2-\frac{1}{m}$.
	
	Equation \eqref{eq:mb} basically states that the expected number of items that $\mechb$ allocates to $v$ is at least $|R_v|$ times the expected fraction of items  that are not in $X$,  divided by $n$.
	
	\begin{lemma}\label{lem:more_x}
		For every agent $v$, if there exists an agent $i$ such that $R_v\neq R_i$ then
		\begin{equation}
		E[\hat{f}_v(\alx_v)] \geq \frac{1}{m^2n(n-1)}+\frac{|R_v|(2m-1)}{nm^2}. \label{eq:expx}
		\end{equation}
	\end{lemma}
	\begin{proof}
		Let $\alpha \stackrel{def}{=} |X|$.
		
		
		By considering the two possible values of $\alpha$ we observe that:
		\begin{equation}
		E[\hat{f}_v(\alx_v)]   =  
		\frac{1}{m} E[\hat{f}_v(\alx_v)~ |~ \alpha=1]+  
		\frac{m-1}{m} E[\hat{f}_v(\alx_v)~|~\alpha=2]. \label{eq:alpha}
		\end{equation}
		Similarly to Equation \eqref{eq:mb} for a fixed $\alpha'\in \{1,2\}$ we have that 
		\begin{equation}
		E[\hat{f}_v(\alx_v) \mid \alpha = \alpha'] \geq  \frac{E[|R_v \cap X| \mid \alpha =\alpha']}{n} = \frac{\alpha'\cdot |R_v|}{nm}, \label{eq:prop2}
		\end{equation}
		where Equation \eqref{eq:prop2} is by conditioning on the size of $R_v \cap X$. If $|R_v\cap X |=1$ and if $v$ is the highest priority agent in $\mechx$ then he gets it. If $|R_v\cap X |=2$ and $v$ among the two highest original priority agents in $\mechx$ then he gets at least one item.

		We consider two cases of $R=(R_1,\ldots,R_n)$:
		
		\textbf{Case 1:}
		There exist an agent $i$ and item $j$ such that $j\in R_v \setminus R_i$.
		In this case if $X=(j)$ then the probability over the priorities that agent $v$ will get item $j$ is at least $\frac{1}{n-1}$, thus:
		\begin{eqnarray}
		E[\hat{f}_v(\alx_v)\mid \alpha=1]& = & E[\hat{f}_v(\alx_v)\mid X=(j)\wedge \alpha=1]\cdot Pr[X=(j) \mid \alpha=1] +\nonumber \\
		& &  E[\hat{f}_v(\alx_v)\mid X\neq (j) \wedge\alpha=1]\cdot Pr[X \neq (j) \mid \alpha=1] \nonumber \\
		& \geq & \frac{1}{m}\cdot\frac{1}{n-1} + \frac{m-1}{m}\cdot \frac{|R_v|-1}{n(m-1)} = \frac{1}{mn(n-1)}+\frac{|R_v|}{nm},\label{eq:alpha3}
		\end{eqnarray}
		since for every $a\in R_v\setminus \{j\}$ if $X=(a)$ then $a$ is allocated to $v$ with probability of at least $\frac{1}{n}$ while if $X=(j)$, $j$  is allocated to $v$ with probability at least $\frac{1}{n-1}$.
		
		By combining Equations \eqref{eq:alpha},\eqref{eq:prop2},\eqref{eq:alpha3}, we get that:
		\begin{eqnarray*}
			E[\hat{f}_v(\alx_v)] & \geq & 
			\frac{1}{m}\left(\frac{1}{mn(n-1)}+\frac{|R_v|}{nm}\right)+
			\frac{m-1}{m}\cdot\frac{2\cdot |R_v|}{nm}\nonumber \\
			&= &\frac{1}{m^2n(n-1)}+\frac{|R_v|(2m-1)}{nm^2}.
		\end{eqnarray*}

		%
		\textbf{Case 2:}
		There exist an agent $i$ and an item $j$ such that $j\in R_i \setminus R_v$.
		In this case if $X=(j,\ell)$
		for $\ell \in R_v$, then the probability over the priorities that agent $v$ will get $\ell$ is at least $\frac{1}{n-1}$, thus:
		
		\begin{eqnarray}
		E[\hat{f}_v(\alx_v) \mid \alpha=2] & = & E[\hat{f}_v(\alx_v)\mid X=(j,\ell)\wedge \alpha=2]\cdot Pr[X=(j,\ell) \mid \alpha=2] +\nonumber \\
		& &  E[\hat{f}_v(\alx_v)\mid 2=|X\cap R_v| \wedge\alpha=2]\cdot Pr[|X\cap R_v| =2  \mid \alpha=2] +\nonumber   \\
		& & E[\hat{f}_v(\alx_v)\mid X\neq(j,\ell) \wedge |X\cap R_v|=1 \wedge \alpha=2]\times \nonumber \\ & & Pr[X\neq (j,\ell) \wedge |X\cap R_v|=1\mid \alpha=2] \nonumber \\
		& \geq  &  \frac{1}{n-1} \cdot \frac{1}{m(m-1)} + \nonumber \\
		&   &  \frac{2}{n} \cdot \frac{|R_v|(R_v|-1)}{m(m-1)} + \nonumber \\
		&   &  \frac{1}{n} \cdot \frac{2\cdot |R_v| (m-|R_v|)-1}{m(m-1)}  \nonumber \\
		& = & \frac{1}{m(m-1)n(n-1)} +\frac{2|R_v|}{mn},\label{eq:alpha4}
		\end{eqnarray}
		where the inequality holds since  if $X = (j,\ell)$ then the expected number of items agent $v$ gets is at least $\frac{1}{n-1}$. 
		If $X \neq (j,\ell)$ then the expected number of items agent $v$ gets is at least $\frac{|R_v \cap X|}{n}$.   

		By combining Equations \eqref{eq:alpha},\eqref{eq:prop2},\eqref{eq:alpha4}, we get that:
		\begin{eqnarray*}
			E[\hat{f}_v(\alx_v)] & \geq & 
			\frac{1}{m}\cdot \frac{|R_v|}{nm}+
			\frac{m-1}{m} \left(\frac{1}{m(m-1)n(n-1)} +\frac{2\cdot |R_v|}{nm}\right)\nonumber \\
			&= &\frac{1}{m^2n(n-1)}+\frac{|R_v|(2m-1)}{nm^2}. \label{eq:}
		\end{eqnarray*}
	\end{proof}
	
	With Lemma \ref{lem:more_x}, we have
	\begin{eqnarray*}
		E[f_v(A_v)] 
		& \stackrel{\eqref{eq:dichotomous_lower}}{\geq} & E[\hat{f}_v(A_v)]  \nonumber\\
		& = & E[\hat{f}_v(\alb_v)] + E[\hat{f}_v(\alx_v)] \\ 
		&\stackrel{\eqref{eq:mb}\eqref{eq:expx}}{\geq}& \frac{|R_v|}{n}\cdot \frac{m^2-2m+1}{m^2} +\frac{1}{m^2n(n-1)}+\frac{|R_v|(2m-1)}{nm^2} \\
		& \stackrel{}{> } & \frac{|R_v|}{n}+\frac{\epsilon \cdot m}{n} \geq \frac{(1+\epsilon)|R_v|}{n} \stackrel{\eqref{eq:proportional_upper}}{\geq}\frac{f(R_v)}{n},
		\label{eq:proportional_lower}   
	\end{eqnarray*}
	which is the proportional share of agent $v$.

\end{proof}
		
\begin{lemma}\label{lem:level-polynomial}
$\meche$ can be implemented in  polynomial time in the number of items and agents.	
\end{lemma}		
\begin{proof}		
Agents are asked to report to $\meche$ only their demand sets. Consequently, the computational  complexity of $\meche$ does not depend of the number of bits needed to represent the valuations, and only depends on the number of agents and items. Both $\mech^{B}$  and $\mechx$ can be implemented in time polynomial in the number of items and agents.  Other steps of $\meche$ include selecting a random permutation over the agents, and selecting at random one or two items to be included in $X$, and these steps can also be implemented efficiently (given a source of randomness). It follows that $\meche$ can be implemented in polynomial time.  
\end{proof}

\cout{
\subsection{$\meche$ is not EFX for Additive Dichotomous Valuations}
\begin{example}[$\meche$ is not EFX for additive dichotomous valuations]\label{example:not-EFX}
	Consider setting with $2$ items, and two players with additive dichotomous valuations. Player 1 desires only item $1$ while player 2 desires both of the items.
	In case where $X=(1)$ and player 1 with the higher priority, player 2 will get both of the items, and this allocation is not EFX.
\end{example} 
}

\section{XOS valuations}\label{sec:XOS}
We next consider dichotomous valuations beyond the submodular case. A class of valuations that contains submodular valuations is the class of XOS valuations.
An XOS valuation $f$ is defined by a set of additive valuations $\{f_1,\ldots,f_k\}$ and for every $S$, $f(S)=\max_{i\in[k]}f_i(S)$. An XOS dichotomous valuation, is a function that is both XOS and dichotomous.  
In this appendix we present some negative results regarding what can be achieved in XOS markets. We use the following construction of an XOS dichotomous valuation. Given a 
family $F$ of sets of items, we define $f_F(S)= \max_{T \in F} |T\cap S|$. Clearly $f_F$ is XOS and dichotomous, 
since we can define for every $T \in F$, the additive function $f_T(S)=|T\cap S|$, and $f_F$ is the max over the $\{f_T\}_{T\in F}$. 

We use such valuations to show  
that it is not possible to extend our result for submodular dichotomous agents to XOS dichotomous agents, even if there are only two agents and only four items. We show that truthfulness and welfare maximization are at odds, even if one disregards all fairness considerations. 
This holds not only for deterministic truthful mechanisms, but even for randomized mechanisms that are only required to be truthful in expectation, as long as the mechanism must still maximize welfare ex-post. 

\begin{proposition}
	For the setting with two dichotomous XOS agents and four items, there is no randomized truthful-in-expectation mechanism 
	that always maximizes welfare.
\label{prop:DownwardClosed_constraints}
\end{proposition}
\begin{proof}
Let the set of items $M=\{1,2,3,4\}$.
Given a family $F$ of feasible subsets of $M$, 
	let $f_F$ be the XOS dichotomous function $$f_F(S)= \max_{T \in F} |T\cap S|$$	
Consider any mechanism that always picks an allocation that maximizes the welfare.
If both agents have the same family $F_1$ with only one feasible set $T =  \{2,3,4\}$,
then there is an agent that gets more than one item in expectation, as welfare maximization implies that all three items in $\{2,3,4\}$ must be allocated.
W.l.o.g., we assume that agent~1 is that agent. Suppose now that agent~1 has the family $F_2$, that contains $T$ and the set $\{1\}$, 
then if agent~1 reports $F_2$ (and agent~2 reports $F_1$), a welfare maximizing mechanism must allocate item~1 to agent~1 and the remaining items to agent~2. 
Yet agent~1 can get higher expected value by reporting $F_1$, and thus the mechanism is not truthful in expectation.
\end{proof}

{
\begin{remark}
	The same arguments of the proof of Proposition~\ref{prop:DownwardClosed_constraints} shows that similarly there is no randomized truthful-in-expectation mechanism that always maximizes the Nash social welfare.
\end{remark}
}

We next show that every truthful mechanism cannot give a $\sqrt{3}-1$ approximation to the optimal social welfare.
\begin{theorem}
	For any $c>\sqrt{3}-1$, there is no truthful deterministic allocation mechanism for two dichotomous XOS agents that always gives a $c$-fraction of the maximal welfare.
\end{theorem}
\begin{proof}
	Assume towards contradiction that there is such a mechanism.
	Consider a market with two agents and two (large enough) disjoint sets of items $A,B$ as we define below.
	We define $F_{A}=\{A\},F_{AB}=\{A,B\}$.  
If both agents report their valuation is $f_{F_A}$, by welfare approximation at least $c\cdot |A|$ are allocated, and w.l.o.g., assume that agent~$1$ receives at least $\frac{c}{2} \cdot |A|$.
If $|B| = \frac{c}{2} \cdot |A| -1$, then if agent 1 reports his valuation is $f_{F_{AB}}$, instead of $f_{F_A}$, he cannot receive items from $B$ (since it will decrease his welfare). Thus, the welfare is bounded by $|A|$ although the maximal social welfare is $|A|+|B|$.
Since  $c \cdot (|A|+ \frac{c}{2} \cdot |A| -1) > |A|$, (for large enough $|A|$) this leads to a contradiction.
\end{proof}

We next show that for dichotomous XOS valuations there are instances with $n$ items in which 
	the maximum welfare is almost factor $2-1/n$ larger than the welfare of any EF1 allocation.
\begin{theorem} \label{thm:SW-EF1}
	There exists a setting with $n$ agents with   dichotomous XOS valuations in which the maximal social welfare is $2nk+n-k$,
	while any EF1 allocation has social welfare at most $n(k+1)$.
	Thus, for any fixed $n$, for any $\epsilon>0$,  
	the ratio of the maximal welfare to the  welfare of any EF1 allocation is at least  $2-1/n- \epsilon$ when $k$ is large enough. 
\end{theorem}
\begin{proof}
	Fix a large k. Let $S_1,\ldots,S_k$ be disjoints sets of items such that $|S_1|=n(k+1)$, and $|S_j|= k$ for $2\leq j\leq n$. Let $F_1=\{S_1\}$, and let $F_j=\{S_1,S_j\}$ for  $j>1$, and the valaution of agent $j$ is $f_{F_j}$. 
	The maximal social welfare is obtained  by giving each agent $j$ the set $S_j$, yielding a  welfare of  $n(k+1) + (n-1)k =2nk+n-k$. 
	When considering  an EF1 allocation, either it only allocates items from $S_1$. In this case the welfare is bounded by $n(k+1)$. Else there exists an agent $j\neq1$ that gets items from $S_j$. In this case, by the EF1 property no agent can get more than $k+1$ items from $S_1$, therefore their value cannot exceed $k+1$. Thus, the welfare cannot exceed $n(k+1)-1$. 
	
The ratio of the maximal welfare to the maximal welfare of an EF1 allocation is $(2nk+n-k)/((k+1)n) =(2k+1)/(k+1)-k/((k+1)n) $ and it tends to $2-1/n$ as $k$ grows large. 
\end{proof}

We now consider the notion of Nash Social Welfare (NSW) and show that for dichotomous XOS valuations, 
the maximum NSW can be about twice as large as the NSW of any EF1 allocation. 

\begin{theorem}
	There exists a setting with $n$ agents with   dichotomous XOS valuations in which any EF1 allocation has Nash Social Welfare of at most $k/2+1$, while the maximal NSW is  $k \cdot (1/4)^{1/n}$. Thus,
	for any $\epsilon>0$, the ratio of the maximal NSW to the NSW of any EF1 allocation is at least $2-\epsilon$ when $n$ is large enough.
\end{theorem}
\begin{proof}
	Let $k>>n$ be a large even number.
	Consider the following setting with $n$ agents and $k(n-1)$ items.
	Let $S_j$ for $1\leq j \leq n-1$ be disjoint sets of size $k$, and let $F=\{S_j \mid 1\leq j \leq n-1\}$.
Each agent valuation is $f_F$.
The maximal NSW is obtained by giving every agent $j\leq n-2$ the set $S_j$, and split $S_{n-1}$ equally between agent $n-1$ and agent $n$.  This allocation has a NSW of $(k^{n-2} \cdot (k/2)^2)^{1/n} = k \cdot (1/4)^{1/n}$.
On the other hand, in any allocation at least one agent gets at most $k/2$ items, so in any  EF1 allocation, every agent gets at most $k/2+1$ items, so the NSW an EF1 allocation is bounded by $((k/2+1)^n)^{1/n} = k/2+1$.
\end{proof}

We next show that for the setting of dichotomous XOS, truthful allocation mechanisms cannot maximize the NSW (and constant fraction of the welfare will sometimes be lost).
\begin{theorem}
	For any $c>2^{-1/3}$, there is no truthful deterministic allocation mechanism for two XOS dichotomous agents that always gives a $c$-fraction of the maximal Nash social welfare.
\end{theorem}
\begin{proof}
	Assume towards contradiction that there is such a mechanism.
	Consider a market with two agents and two (large enough) disjoint sets of items $A,B$ as we define below.
	Let $F_{A}=\{A\}$, and let $F_{AB}=\{A,B\}$.
If both agents report their valuation is $f_{F_A}$, by assumption the allocation NSW must be such that at least $c\cdot |A|$ items are allocated, and w.l.o.g., assume that agent $1$ receives at least $\frac{c}{2} \cdot |A|$.
If $|B| = \frac{c}{2} \cdot |A| -1$, then if agent 1 reports his valuation is $f_{F_{AB}}$, instead of $f_{F_A}$, he cannot receive items from $B$ (since it will decrease his welfare). Thus, the Nash social welfare is bounded by $|A|/2$ although the maximal social welfare is $\sqrt{|A||B|}$.
	Since  $c \sqrt{|A||B|} \approx c\cdot \sqrt{\frac{c}{2}} \cdot |A| > |A|/2$, (for large enough $|A|$) this leads to a contradiction. 
\end{proof}

\section{Best-of-Both-Worlds via Random Priorities} 
\label{app:RPE}
{For completeness, we first present a proof of a simple observation for agents with subadditive valuations: a randomized ex-ante envy-free allocation that is supported on (deterministic) allocations that are non-wasteful (for each agent, the marginal value of the set on unallocated items is zero) is also  ex-ante proportional (each agent gets at least $\frac{1}{n}$ of his value for the grand bundle).}
\begin{observation}\label{obs:ex-ante-EF-prop}
	For the class of for subadditive valuation functions, 
	any ex-ante envy-free randomized allocation that is not wasteful is also ex-ante proportional. 
\end{observation}
\begin{proof}
	Assume that  $A$ is the (random) allocation and let $A_0$ be the (random) set of unallocated items. Let $f_i$ be the valuation function of agent $i$. For any realization of $A$, as the realized allocation is not wasteful it holds that $f_i(A_i \cup A_0)=f_i(A_i)$, and by subadditivity, $f_i(A_i \cup A_0) + \sum_{j\neq i } f_i(A_j) \geq f_i(M)$.
	By ex-ante envy-freeness, $E[f_i(A_i)] \geq E[f_i(A_j)]$ for every  $i, j$, and thus $E[f_i(A_i)] \geq \frac{1}{n}  \left(E[f_i(A_i)] + \sum_{j\neq i}E[f_i(A_j)]\right)$.	
	Combining these three equations imply the claim:  $$E[f_i(A_i)] \geq \frac{1}{n} \left(E[f_i(A_i)] + \sum_{j\neq i}E[f_i(A_j)] \right) =  \frac{1}{n} \left(E[f_i(A_i \cup A_0)] + \sum_{j\neq i}E[f_i(A_j)] \right)   \geq \frac{1}{n} f_i(M).
		$$
\end{proof}

We next show that our RPE mechanism is stochastically envy free, which together with our results for the PE mechanism, immediately imply Theorem~\ref{thm:RPE-main}.

\begin{lemma}	\label{lem:EV}
	For the class of submodular dichotomous valuations functions, the random priority egalitarian (RPE) mechanism  
	is stochastically envy-free.
	This implies that it is also ex-ante envy free and ex-ante proportional for this class.
\end{lemma}
\begin{proof}
For the purpose of following the proof of this theorem, the reader is advised to recall Theorem~\ref{thm:LorenzDuttaRay} and the notion of auxiliary items that precedes it.

Consider two agents $i$ and $j$, and two priority orders $\pi$ and $\pi'$, where $i$ appears before $j$ in $\pi$, and $\pi'$ is identical to $\pi$, except that $i$ and $j$ switch locations.
When allocating both items as well as the auxiliary items, let $\hat{A}$ ($\hat{A}'$, respectively)  be a Lorenz dominating allocation under priority order $\pi$ ($\pi'$, respectively), and let $A$ ($A'$, respectively) be the corresponding allocations when the auxiliary items are removed.
Recall that given $\pi$, all Lorenz dominating allocations have the same sorted vector of utilities. Moreover, Lorenz domination implies that (ignoring auxiliary items) $A$ and $A'$ also have the same vector of utilities, up to permuting the names of the agents. As $i$ appears before $j$ in $\pi$, it implies that in allocation $A$, agent $i$ cannot envy $j$, but $j$ might envy $i$. Suppose that $j$ does envy $i$ under $A$, that is, $v_j(A(i)) > v_j(A(j))$. Then we show below that $|A'(j)| \ge |A(i)|$ and that $|A(j)| \ge |A'(i)| $. Lemma~\ref{lem:EV} follows easily from the claim, and the fact that $ v_j(A(j)) = |A(j)|$ and $ v_j(A'(j)) = |A'(j)|$, since the RPE mechanism never allocates undesired items, while $ v_j(A(i))\le |A(i)| $ and $ v_j(A'(i)) \le |A'(i)|$ as valuations are dichotomous.

We now prove the claim.  As the RPE mechanism never allocates undesired items, $v_j(A(j))=|A(j)|$. 
The fact that $j$ envies $i$ in $A$ together with the EFX property implies that $|A(i)| = |A(j)| + 1$, and that $v_j(A(i)) = |A(i)|$. The matroid exchange property implies that  
we can move one item from $A(i)$ to $A(j)$, and by this get an allocation $B$ in which $v_j(B(j)) = v_j(A(j)) + 1 = |A(i)|$ and $v_i(B(i)) = |B(i)| = |A(i)| - 1 = |A(j)|$. As $B$ has the same utility vector as $A$ (up to permuting the names of agents), $B$ is a Lorenz dominating allocation.
As before, we use $\hat{B}$ to denote the allocation that is identical to $B$ but with each agent also getting his  auxiliary item. If $\hat{B}$
is a Lorenz dominating allocation with respect to priority order $\pi'$, then the claim is proved. Hence we assume that $\hat{A}'$ strictly Lorenz dominates $\hat{B}$ 
with respect to $\pi'$. (Example~\ref{ex:A'domB} shows that indeed it may happen that
$\hat{A}'$ strictly Lorenz dominates $\hat{B}$  
with respect to $\pi'$.)

Consider $\pi'$ (which defines auxiliary items as in Theorem~\ref{thm:LorenzDuttaRay}), and for each of the allocations
$\hat{A}'$ and $\hat{B}$ 
order the utilities of agents from smallest to largest. Call the first of these sorted vectors
$L_{\hat{A}'}$ and the other $L_{\hat{B}}$.  
When auxiliary items are removed, we refer to these vectors as $\lfloor L_{A'} \rfloor$ and $\lfloor L_{B} \rfloor$.
Note that these vectors are still sorted.  Similarly, we define the vectors $L_{\hat{A}}$ and $\lfloor L_{A} \rfloor$ for the allocation $\hat{A}$ with auxiliary items as in $\pi$. Observe that $L_{\hat{A}} = L_{\hat{B}}$, which holds because $(\pi,\hat{A})$ and $(\pi',\hat{B})$ only interchange the roles of $i$ and $j$ (both the auxiliary items are interchanged, and the number of allocated items to each agent are interchanged).

As $A'$ and $B$ are both Lorenz dominating allocations, then $\lfloor L_{A'} \rfloor$ and $\lfloor L_{B} \rfloor$ are identical. However, with the auxiliary items, $L_{\hat{A}'}$ strictly Lorenz dominates $L_{\hat{B}}$. Let $P$ ($S$, respectively) be the set of agents in the longest prefix (suffix, respectively) in which $L_{\hat{A}'}$ and $L_{\hat{B}}$ are identical. (The utility identifies the agent, because of the auxiliary items.)

{\bf Proof that $|A(j)| \ge |A'(i)|$.}
If $i \in P$ then $|A'(i)| = |B(i)| = |A(j)|$, as desired. If $i \not\in P$ then also $j \not\in P$ (because $|B(j)| > |B(i)|$). Let $p$ be the agent following $P$ in the sorted vector for $B$, and let $p'$ be the agent following $P$ in the sorted vector for $A'$. We refer to $P \cup \{p'\}$ as the extended prefix of $A'$ (and likewise, $P \cup \{p\}$ is the extended prefix for $B$).

\begin{enumerate}

    \item If $p' = i$ and $p \not= i$ then $|A'(i)| \le |B(i)|$ (otherwise $\lfloor L_{A'} \rfloor$ and $\lfloor L_{B} \rfloor$ cannot be identical), implying $|A(j)| \ge |A'(i)|$, as desired.

        \item The case $p' = p = i$ cannot hold because strict Lorenz domination would imply that $|A'(i)| > |B(i)|$. This contradicts the fact that $\lfloor L_{A'} \rfloor$ and $\lfloor L_{B} \rfloor$ are identical.

         \item If $p' \not= i$ then consider the allocation $\hat{A}'$, but with auxiliary item values set from $\pi$ instead of $\pi'$. 
         The fact that $\hat{A}$ Lorenz dominates $\hat{A}'$ with respect to $\pi$ (together with $L_{\hat{A}} = L_{\hat{B}}$) implies that $L_{\hat{A}'}$ for $\pi$ has a different extended prefix compared to $L_{\hat{A}'}$ for $\pi'$. This can only happen if $i$ is moved into the extended prefix (this will happen if $i$ and $p'$ receive the same number of items under $A'$, and $i$ precedes $p'$ in $\pi$ but not in $\pi'$), but then $|A(j)| \ge |A'(i)|$, as desired.


\end{enumerate}

The combination of the above cases establishes that $|A(j)| \ge |A'(i)|$.

{\bf Proof that $|A'(j)| \ge |A(i)|$.} (This proof is analogous to that for $|A(j)| \ge |A'(i)|$, but we present it in full for completeness.)
If $j \in S$ then $|A'(j)| = |B(j)| = |A(i)|$, as desired. If $j \not\in S$ then also $i \not\in S$ (because $|B(j)| > |B(i)|$). Let $s$ be the agent preceding $S$ in the sorted vector for $B$, and let $s'$ be the agent preceding $S$ in the sorted vector for $A'$.  We refer to $S \cup \{s'\}$ as the extended suffix of $A'$ (and likewise, $S \cup \{s\}$ is the extended prefix for $B$).

\begin{enumerate}

    \item If $s' = j$ and $s \not= j$ then $|A'(j)| \ge |B(j)|$ (otherwise $\lfloor L_{A'} \rfloor$ and $\lfloor L_{B} \rfloor$ cannot be identical), implying $|A'(j)| \ge |A(i)|$, as desired.

    \item The case $s' = s = j$ cannot hold because strict Lorenz domination would imply that $|A'(j)| < |B(j)|$. This contradicts the fact that $\lfloor L_{A'} \rfloor$ and $\lfloor L_{B} \rfloor$ are identical.

      \item If $s' \not= j$ and $s \not= j$, then consider the allocation $\hat{A}'$, but with auxiliary item values set from $\pi$ instead of $\pi'$. The fact that $\hat{A}$ Lorenz dominates $\hat{A}'$ with respect to $\pi$ (together with $L_{\hat{A}} = L_{\hat{B}}$) implies that $L_{\hat{A}'}$ for $\pi$ has a different extended suffix compared to $L_{\hat{A}'}$ for $\pi'$. This can only happen if $j$ is moved into the extended suffix (this will happen if $j$ and $s'$ receive the same number of items under $A'$, and $j$ precedes $p'$ in $\pi'$ but not in $\pi$), but then $|A'(j)| \ge |A(i)|$, as desired.


\end{enumerate}

The combination of the above cases establishes that $|A'(j)| \ge |A(i)|$.
\end{proof}


In the next example, we show that allocation $B$ in the proof of Lemma~\ref{lem:EV} might not be Lorenz dominating with respect to $\pi'$.
\begin{example}
\label{ex:A'domB}
Suppose that there are four agents $\{1,2,3,4\}$ and six items $\{a,b,c,d,e_1,e_2\}$. 
The sets of items desired by the agents are $(a,e_1,e_2)$, $(a,b)$, $(c,e_1,e_2)$, $(a,d,e_1,e_2)$, respectively. Valuation functions are additive, except for one exception, which is that for agent~1 items $e_1$ and $e_2$ are substitutes of each other. In every Lorenz dominating allocation two agents get a pair of items each, and two agents get one item each.

For priority order $\pi = (1,2,3,4)$, a Lorenz dominating allocation $A$ is $\{(a,e_1), (b), (c,e_2), (d)\}$. In this allocation agent~4 envies agent~1. For the permuted priority order $\pi' = (4,2,3,1)$, allocation $B$ would leave the bundles allocated to agents~2 and~3 as they were in $A$. However, the allocation $A'$ that is Lorenz dominating for $\pi'$ is $\{(d,e_2), (a, b), (c), (e_1)\}$. In this allocation agent~2 gains an item and agent~3 loses an item. Hence $B$ is not a Lorenz dominating allocation with respect to $\pi'$.
\end{example}

	In the following example we show the difference between ex-ante envy-freeness and stochastic envy-freeness.
\begin{example} \label{ex:ex-ante-EV-vs-stoc}
	Consider setting with five items $(a,b,c,d,e)$ and three additive dichotomous agents, where agent~1 desires $(a,b,c)$, whereas the other two agents desire $(a,b,c,d,e)$. The fractional allocation in which agent~1 gets item $a$ and agents~2 and~3 get half of each of the remaining items is ex-ante EF. If this fractional allocation is rounded in such a way that agent~2 gets items $(b,c)$ with probability $\frac{1}{2}$ and items $(d,e)$ and with probability $\frac{1}{2}$ (agent~3 gets the remaining items), then  ex-post, this allocation is EFX. However, the randomized allocation is not stochastically EF: agent~1 never gets value more than~1, but other agents do get bundles that are worth to agent~1 more than~1 with positive probability.
\end{example}


\section{Allocating All Items}
\label{app:all}

In Theorem~\ref{thm:submodular-main-intro} we present a truthful deterministic allocation mechanism that is welfare maximizing and EFX, but leaves all undesired items unallocated (non-redundant).
One might wonder if it is possible to always allocate all items and obtain a similar result.
We next show that the result of Theorem  \ref{thm:submodular-main-intro} is impossible to obtain when one insists on allocating all items (even undesired ones).
Specifically, we show that there is no truthful deterministic allocation mechanism that always allocates all items, maximizes welfare and is EFX. Moreover, this impossibility holds even for additive dichotomous agents, and even for two agents.

\begin{theorem}
	For 2 additive dichotomous agents, there is no truthful deterministic allocation mechanism that is welfare maximizing, EFX, and always allocates all items.
\end{theorem}
\begin{proof}
Consider a market with $m\geq 21$ items.
When both agents report they want the same set $X$ of size 3 from $M$, 
by EFX and welfare maximizing, we get that the allocation must be such that one agent receives two items among $X$ and the other receives the third and all items not in $X$.
By a counting argument we get that there are $6$ items $a,b,x_1,x_2,x_3,x_4$, such that for every  $X=\{a,b,x_i\}$ for $i\in [4]$ the same agent receives exactly the set $\{a,b\}$ and the other agent receives $M\setminus\{a,b\}$.
This is true since there are ${m\choose 3}$ sets 
of size 3, and only $2 \cdot {m\choose 2}$ allocations that gives one agent 2 items and the other, the rest of the items. Since for $m\geq 21$ it holds that ${m\choose3} > 3\cdot 2 \cdot {m\choose2}$, we get that there are four sets of size 3 with the same allocation.
We assume w.l.o.g., that agent 1 is the agent that receives $a,b$ whenever both agents reports $\{a,b,x_i\}$.

By truthfulness and welfare maximizing if agent 2 changes his report to $M\setminus\{a,b\}$, agent 1 receives $\{a,b\}$ and agent 2 receives $M\setminus\{a,b\}$.
If agent $1$ reports he wants $\{a,b,x_1,x_2,x_3,x_4\}$ and agent $2$ reports he wants $M\setminus\{a,b\}$ then by EFX agent 1 must receive at least $3$ items and $a,b$ among them. Contradicting that it is truthful for him to report $\{a,b,x_i\}$ for every $i \in [4]$.
\end{proof}

\if

There exist items $a,b,x_1,\ldots x_4$ such that whenever both players want the set  $\{a,b,x_i\}$ agent 1 gets $a,b$.

\begin{itemize}
	\item If $|X\cap f(X,X) |> \frac{|X|}{2}$ then $|X \cap Y \cap f(X,X)| \geq |X \cap Y \cap f(X,Y)|$.
	\item If $|X\cap f(X,X)| < \frac{|X|}{2}$ then $|X \cap Y \cap f(X,X)| \leq |X \cap Y \cap f(Y,X)|$.
\end{itemize}

\begin{theorem}
	For 2 additive dichotomous agents, there is no truthful deterministic allocation mechanism that allocates always all items for $m > TBD$ items, and is welfare maximizer and EFX.
\end{theorem}
\begin{proof}
Let $f(S_1,S_2)$ be the allocation to agent 1 ($M\setminus f(S_1,S_2)$ is the allocation to agent 2).
If $|f(X,X)| > \frac{|X|}{2}$ then 
\end{proof}

\begin{claim}
	approximating WF (without normalization) gives bad approximation to fairness since if there is one agent with values much higher than all others, and each agent deserve one item, 
\end{claim}

If we consider deterministic EFX welfare maximizing truthful allocating mechanisms, we will also have the following properties:
 
Let $f(S_1,S_2)$ be the allocation to agent 1 ($M\setminus f(S_1,S_2)$ is the allocation to agent 2).

$f(f(M,M),M) =f(M,M)$ by truthfulness.

$|f(M,M) \cap f(f(M,M),f(M,M))|  \leq   \frac{|f(M,M)|}{2} $, otherwise by EFX $M\setminus f(f(M,M),f(M,M))$ is given to agent 2, which is greater then $M\setminus f(M,M)$ which is given to agent 2 when reporting $M$, contradicting truthfulness.

\begin{corollary}
	There is no prioritized such mechanism.
\end{corollary}
\begin{proof}
Consider a market with 2 agents and 6 items $M = \{1,\ldots,6\}$.
	W.l.o.g., let agent 1 be the one that agent with the higher priority.

	W.l.o.g. $f(M,M)=\{1,2,3\}$, then by the former claim $|f(\{1,2,3\},\{1,2,3\}) \cap \{1,2,3\} |=1$
\end{proof}

For $|A|>1$ if $|f(A,A)| > \frac{|A|}{2}$ then for every $B$, $f(B,B) \neq A$, otherwise $f(A,B)=A$ and then $B \setminus f(A,A)$ better than  $B \setminus f(A,B)$, contradicting truthfulness of agent 2.
Similarly if $|f(A,A)| < \frac{|A|}{2}$ then for every $B$, $f(B,B) \neq B\setminus A$.

\fi

\section{Group Strategyproofness}
\label{sec:group-SP}

In this appendix we consider the aspect of group strategyproofness of our PE and RPE allocation mechanisms.
{A mechanism is \emph{strongly group strategyproof} if there is no deviation by a group in which every member of the group weakly gains, and one strictly so.
A mechanism is \emph{weakly group strategyproof} if there is no deviation by a group in which every member of the group strictly gains.

~\citet{BM2004} prove that for 
unit-demand dichotomous valuations, no Pareto optimal deterministic allocation mechanism  is  strongly group strategyproof. As the class of submodular dichotomous valuations that we consider contains unit-demand valuations, 
strong group strategyproof is not obtainable in our setting. For randomized allocation we  have the following impossibility:  

\begin{observation}
No randomized allocation mechanisms for submodular dichotomous valuations can be simultaneously ex-ante strongly group strategyproof, ex-post Pareto optimal and EF1.
\end{observation}
\begin{proof}
Consider a setting with three agents $a_1, a_2, a_3$ with dichotomous valuations and four items $e_1, e_2, e_3, e_4$. Agent $a_1$ is additive over $e_1,e_2$, agent $a_2$ is additive over $e_3,e_4$, and agent $a_3$ is unit demand over all items. In every EF1 PO allocation, all items are allocated, and $a_3$ gets exactly one item. In a randomized allocation mechanism that is ex-post EF1 and PO, w.l.o.g., $a_3$ has positive probability of receiving item $e_1$. Then the group $\{a_1,a_3\}$ has a deviation that benefits $a_1$ without hurting $a_3$, and this is for $a_3$ to report that $a_3$ is unit demand over $e_3,e_4$. With this deviation, 
in every {realized allocation that is EF1 and PO,} 
agent $a_1$  gets both $e_1$ and $e_2$ (sometimes gaining an item), whereas agent $a_3$ still gets one item {in that realization (either $e_3$ or $e_4$) and losses nothing}. 
\end{proof}

We next discuss the weaker notion of weak group strategyproofness, as well as the  properties of our mechanisms for the additive case. For additive dichotomous valuations, our PE mechanism is weakly group strategyproof, as shown by \citet{halpern2020fair}.} 
The next example shows that our PE allocation mechanism (and likewise the mechanism of \citet{halpern2020fair}) is not strongly group strategyproof, even if the dichotomous valuations are additive. 
\begin{example}
	There are three agents, 1,2 and 3,  and two items, $a$ and $b$.
	The agents are ordered according to their priorities (i.e., agent 1 has higher priority than agents 2 and 3, and agent 2 has higher priority than agent 3).
	Agent 1 wants both items $a,b$, while agent 2 wants only item $a$, and agent 3 wants item $b$.
	If agents $1,3$ collaborated and agent $1$ reports he wants only item $a$, then, agent $3$ gains an item, while agent $1$ does not lose an item.
\end{example}

We next show an example in which our RPE mechanism in Section~\ref{sec:rpe} is neither ex-ante Lorenz dominating, nor {weakly} 
group strategyproof, even if the dichotomous valuations are additive.


\begin{example}
	There are eight agents $a_1, \ldots, a_8$  with dichotomous additive valuations over twelve items $e_1, \ldots e_{12}$.
	Agents $a_1, \ldots ,a_4$ each wants only items $e_1, \ldots , e_6$, whereas agents $a_5, \ldots, a_8$ each wants all twelve items. In any realized allocation, agents $a_5, \ldots ,a_8$ combined get at least six items. The fractional Lorenz dominating allocation gives every agent $\frac{3}{2}$ items. However, when agents $a_5, \ldots, a_8$ have highest priority, they get eight items under PE. Hence under RPE each of them gets average value strictly greater than $\frac{3}{2}$ due to symmetry (and as they always get at least 6 items), whereas each of agents 1 to 4 gets average value strictly smaller than $\frac{3}{2}$. This establishes that RPE is not ex-ante Lorenz dominating.
	
If agents $a_5$ and $a_6$ change their reports to be $e_1, \ldots, e_9$ and agents $a_7$ and $a_8$ change their reports to be $\{e_1, \ldots, e_6\} \cup \{e_{10}, e_{11}, e_{12}\}$ then there is no priority order in which the group $a_5, \ldots, a_8$ gets fewer items than they would get under truthful reporting (with the same priority order). However, on priority order $a_7, a_8, a_1, \ldots a_6$ they get one more item compared to truthful reporting. Symmetry arguments imply that in expectation, any member of the group $a_5, \ldots, a_8$ strictly gains from the deviation above. Hence RPE is not {weakly} 
group strategyproof. 
\end{example}


\cout{
\mbc{I think we have decided  to remove this:}

\begin{example}
	There are three agents $a_1, a_2, a_3$  with dichotomous additive valuations over six items $e_1, \ldots e_{6}$.
	Agents $a_1$ wants only items $e_1, e_2$, whereas agents $a_2, a_3$ each wants all six items. When reporting truthfully, the only allocation is the one that each agent  gets exactly 2 desired items. 	
	If agents $a_2$ and $a_3$ change their reports to be with probability half that agent $a_2$ desires only $e_1,e_2$, while agent $a_3$ reports all six items, and with probability half, agent $a_3$ reports he desires only $e_1,e_2$, while agent $a_2$ desires all six items, then the expected number of items each of agents $a_2,a_3$ will receive is $\frac{5}{2}$, which is strictly better for them.  Hence RPE is not weakly 
	group strategyproof. 
\end{example}

}


\cout{

}

\end{document}